\documentclass[onecolumn,amsmath,amssymb,aps,pra,notitlepage,11pt,tightenlines]{revtex4-1}

\usepackage{setspace}
\usepackage{etoolbox}
\usepackage{dcolumn}
\usepackage{physics}
\usepackage{amsthm}
\usepackage{dsfont}
\usepackage{float}
\usepackage{enumerate}
\usepackage{graphicx}
\usepackage{hyperref}
\usepackage{bm}
\usepackage[depth=subsection]{bookmark}
\usepackage{color}

\newcommand{\STAB}[1]{$\mathrm{STAB}_{#1}$} 
\newcommand{\STABmth}[1]{\mathrm{STAB}_{#1}} 
\newcommand{\SPO}[1]{$\mathrm{SP}_{#1}$} 
\newcommand{\SPOmth}[1]{\mathrm{SP}_{#1}}

\newcommand{\robmag}[2]{\mathcal{R}_{#2}\qty(#1)} 
\newcommand{\lone}[1]{\| #1\|_1} \newcommand{\lonevec}[1]{\lone{\vec{#1}}} \newcommand{\id}{\mathds{1}}     \newcommand{\sign}[1]{\mathrm{sign}\qty(#1)}
\newcommand{\chan}{\mathcal{E}} \newcommand{\RChoi}[1]{\mathcal{R}\qty(\Phi_{#1})}
\newcommand{\RChoiTP}[1]{\mathcal{R}_* \qty(#1)} \newcommand{\Mcap}[1]{\mathcal{C}\qty(#1) }  
\newcommand{\curlyT}{\mathcal{T}} 
\newcommand{\idn}[1]{\id_{#1}} 
\newcommand{\outstate}[3]{\qty(#1 \otimes \idn{#3}) #2} \newcommand{\Rstar}[1]{\RChoiTP{#1}} 
\newcommand{\choistate}[1]{\Phi_{#1}} \newcommand{\CPR}{\mathrm{CPR}} \newcommand{\LofK}[1]{\mathcal{L}\qty(\mathcal{K}#1)} \newcommand{\ketcal}[2]{\ket{\mathcal{#1}#2}}
\newcommand{\Hilb}{\mathcal{H}}

\newtheorem{theorem}{Theorem}[section]

\newtheorem{obs}{Observation}[section]
\newtheorem{lemma}{Lemma}[section]
\newtheorem{conjecture}{Conjecture}[section]

\begin{document}
\floatstyle{ruled}
\newfloat{algo}{htbp}{}
\floatname{algo}{Algorithm }

\title{Quantifying magic for multi-qubit operations}

\author{
James R. Seddon$^{1}$ and Earl T. Campbell$^{2}$}
\affiliation{$^1$Department of Physics and Astronomy, University College London, London, UK}
\email{james.seddon.15@ucl.ac.uk}
\affiliation{$^2$Department of Physics and Astronomy, University of Sheffield, Sheffield, UK}
\date{2nd July, 2019}

\def\thesection{\arabic{section}}

\begin{abstract}
The development of a framework for quantifying ``non-stabiliserness'' of quantum operations is motivated by the magic state model of fault-tolerant quantum computation, and by the need to estimate classical simulation cost for noisy intermediate-scale quantum (NISQ) devices. The robustness of magic was recently proposed as a well-behaved magic monotone for multi-qubit states and quantifies the simulation overhead of circuits composed of Clifford+$T$ gates, or circuits using other gates from the Clifford hierarchy. Here we present a general theory of the ``non-stabiliserness'' of quantum operations rather than states, which are useful for classical simulation of more general circuits. We introduce two magic monotones, called channel robustness and magic capacity, which are well-defined for general $n$-qubit channels and treat all stabiliser-preserving CPTP maps as free operations.  We present two complementary Monte Carlo-type classical simulation algorithms with sample complexity given by these quantities and provide examples of channels where the complexity of our algorithms is exponentially better than previous known simulators. We present additional techniques that ease the difficulty of calculating our monotones for special classes of channels.
\end{abstract}

\maketitle

\section{Introduction}
The Gottesman-Knill theorem showed that circuits comprised of stabiliser state preparations, Clifford gates, Pauli measurements, classical randomness and conditioning can be efficiently simulated by a traditional computer \cite{Gottesman1997,GKtheorem}. If a circuit involves a relatively small proportion of non-Clifford operations, simulation may be within the reach of a classical computer, albeit with a runtime overhead that is expected to scale exponentially with the amount of resource required.

 An important class of devices comprises so-called near-Clifford circuits where simulation may be feasible \cite{Bennink2017,Yoganathan2018}. There are two scenarios where near-Clifford circuits are relevant.  As we enter the era of Noisy Intermediate Scale Quantum (NISQ) devices \cite{Preskill2018}, many experiments proposed as demonstrators of quantum advantage may be near-Clifford so it is important to rigorously understand when a classical simulation is available. Furthermore, in the NISQ regime the need for classical simulation tools for benchmarking and verification becomes more pressing.  The quantification of non-stabiliser resource is also of interest in the context of the magic state model of fault-tolerant quantum computation \cite{Knill2005,Bravyi2005,Campbell2017a}, the second scenario. Any device intended to provide quantum advantage must involve non-stabiliser operations. In circuits employing error-correcting codes, however, it is often not possible for the code to `natively' implement non-Clifford gates fault-tolerantly \cite{Campbell2017a}. Instead, these gates are implemented indirectly by injection of so-called magic states. These are non-stabiliser states that must be prepared using the experimentally costly process of magic state distillation \cite{Knill2005,Bravyi2005,Reichardt2005,Bravyi2012,Campbell2017a,Jones2013,Trout2015,Haah2017,Hastings2018,Krishna2018,Campbell2018,Wang2018}, which is comprised of Clifford-dominated circuits.

Both of these scenarios motivate the development of a resource theory \cite{Vidal1999,Grudka2014,Horodecki2013a,Brandao2015a,Coecke2016,Napoli2016,Stahlke2014,Wang2018,Takagi2018,Regula2018} where the class of free operations is generated by stabiliser state preparations and rounds of stabiliser operations as described above. For the case of odd $d$-dimensional qudits this problem is largely solved by the discrete phase space formalism \cite{Gross2006,Mari2012,Veitch2012,Veitch2014,Ahmadi,Delfosse2017}; odd dimension qudit stabiliser states are characterised by a positive discrete Wigner function. In Ref. \cite{Pashayan2015}, the discrete Wigner function was cast as a quasiprobability distribution, making a direct connection between the negativity of the distribution, and the complexity of calculating expectation values via a Monte Carlo-type simulation algorithm. These techniques have recently been extended to quantify the magic of odd-dimension qudit quantum channels \cite{Wang2019}. However,  the discrete phase space approach cannot be applied cleanly to qubits without excluding some Clifford operations from the free operations~\cite{Delfosse2015,Raussendorf2017}, or losing the ability to compose representations under tensor product \cite{Raussendorf2019}. To retain all multi-qubit stabiliser channels as free operations, then, we must seek alternative approaches.

Howard and Campbell~\cite{Howard2017} introduced a scheme where density matrices are decomposed as real linear combinations of pure stabiliser state projectors. Non-stabiliser states sit outside the convex hull of the pure stabiliser states, so their decompositions necessarily contain negative terms and can again be viewed as quasiprobability distributions, with $\ell_1$-norm strictly larger than 1. The robustness of magic for a state, defined as the minimum $\ell_1$-norm over all valid decompositions, is a monotone under stabiliser operations and has several useful resource-theoretic properties. Alternative approaches include stabiliser rank methods, where the state vector is decomposed as a superposition of stabiliser states \cite{Garcia2014,Bravyi2016,Bravyi2016A,Bravyi2018}. Exact and approximate stabiliser rank, and the associated quantity extent, are measures of magic for pure states. Here we are interested in measures naturally suited for applications to mixed states or general, noisy quantum channels. A stabiliser-based method to simulate noisy circuits by decomposition of states into Pauli operators was recently proposed in Ref. \cite{Rall2019}. 
In this work we characterise the cost of quantum \emph{operations} with respect to the resource theory of magic. Robustness of magic naturally quantifies the cost for a subclass of non-Clifford operations, namely gates from the third level of the Clifford hierarchy. It is less clear how the framework can be extended to more general quantum operations, and formalising this is one of our main aims. 

In Ref. \cite{Bennink2017}, Bennink et al. presented an algorithm in which completely positive trace-preserving (CPTP) maps are decomposed as quasiprobability distributions over a subset of stabiliser-preserving operations that we will call $\CPR$. This subset supplements the Clifford unitaries with Pauli reset channels, in which measurement of some Pauli observable is followed by a conditional Clifford correction, so as to reset a state to a particular +1 Pauli eigenstate. While Bennink et al. showed that $\CPR$ spans the set of CPTP maps, there is no guarantee that all stabiliser-preserving CPTP maps can be found within its convex hull. Indeed, we will see in Section \ref{sec:algorithms} there exist channels that are stabiliser-preserving, but are nevertheless assigned a non-trivial cost by the algorithm of Ref. \cite{Bennink2017}. The implication is that decomposition in terms of elements of $\CPR$ is not the best strategy for simulating general non-stabiliser operations. An obvious extension of Ref. \cite{Bennink2017} is to replace $\CPR$ by the full set of stabiliser-preserving CPTP maps. The technical question to be answered is then how to correctly and concisely represent this set; how can we be sure that we have captured all possible stabiliser-preserving channels? This issue is addressed in Sections \ref{sec:SPO} and \ref{sec:choi}.

In this paper we introduce two magic monotones for channels: the channel robustness $\mathcal{R}_*$ and the magic capacity $\mathcal{C}$. Both are closely related to the robustness of magic for states.  They are well-defined for general $n$-qubit channels and treat all stabiliser-preserving CPTP maps as free operations.  We will see that these monotones give the sample complexity of two classical simulation algorithms. Other magic monotones have been proposed \cite{Veitch2014, Ahmadi, Wang2018} but without known connections to classical simulation algorithms.  Furthermore, we give several examples of channels where the simulation complexities of our approaches are exponentially faster (as a function of gate count) than other quasiprobability simulators such as the Bennink et al. simulator \cite{Bennink2017}. To our knowledge, our algorithms are the first that are known to efficiently simulate all stabiliser-preserving CPTP maps, as opposed to the set of stabiliser operations generated by Clifford gates and Pauli measurements.

 The paper is structured as follows. In Section \ref{sec:defs} we review the properties of robustness of magic and give some definitions. Next, we summarise our main results in Section \ref{sec:results}, before pinning down what we mean by stabiliser-preserving operations in Section \ref{sec:SPO}. Sections \ref{sec:choi} and \ref{sec:cap} are chiefly concerned with proving important properties of our monotones. Two classical simulation algorithms, each related to one of our monotones, are described in Section \ref{sec:algorithms}. Finally, in Section \ref{sec:numerical} we calculate the numerical values of our monotones for operations on up to five qubits, using techniques developed in Appendix \ref{app:diagonal}. MATLAB code to calculate each of our measures is provided at the public repository Ref. \cite{channel_repo}.

\section{\label{sec:defs}Preliminaries}

Let \STAB{n} be the set of $n$-qubit stabiliser states. In an abuse of notation we will use $\ket{\phi} \in \STABmth{n}$ to mean a pure state from this set, and $\rho \in \STABmth{n}$ to mean the density matrix of a state taken from the stabiliser polytope, the convex hull of pure stabiliser states. The pure states in \STAB{n} form an overcomplete basis for the set of $2^n$-dimensional density matrices $\mathcal{D}_n$. We can therefore write the density matrix for \emph{any} 
state as an affine combination of pure stabiliser state projectors 
$\rho = \sum_j q_j \op{\phi_j}$ where $\ket{\phi_j} \in \STABmth{n}$, and $\sum_j q_j = 1$. In general, $q_j$ can be negative. The robustness of magic is defined as the minimal $\ell_1$-norm $\lone{\vec{q}}  = \sum_j \abs{q_j}$ over all possible decompositions \cite{Howard2017}:
\begin{equation}
\robmag{\rho}{} = \min_{\vec{q}} \qty{\lone{\vec{q}} : \sum_j q_j \op{\phi_j} = \rho, \, \ket{\phi_j} \in \STABmth{n} }.
\end{equation}
In the definition above, the state of interest is expressed as a decomposition over pure stabiliser states. By collecting together all terms of the same sign, any state can instead be expressed in terms of a pair of mixed stabiliser states (Figure~\ref{fig:robustness_decomp}). An equivalent definition is then:
\begin{equation}
    \robmag{\rho}{} = \min_{\rho_\pm \in \STABmth{n}} \qty{ 1 + 2p : (1+p) \rho_+ - p \rho_- = \rho, p \geq 0}.
\end{equation}

The robustness of magic is a well-behaved magic monotone, having the following properties:
\begin{enumerate}
    \item \emph{Convexity}: $\robmag{\sum_j q_j \rho_j}{} \leq \sum_j \abs{q_j} \robmag{\rho_j}{}$;
    \item \emph{Faithfulness}: If $\rho \in \STABmth{n}$, then $\robmag{\rho}{}=1$. Otherwise $\robmag{\rho}{}>1$;
    \item \emph{Monotonicity under stabiliser operations}: If $\Lambda$ is a CPTP stabiliser-preserving operation, then $\robmag{\Lambda\qty(\rho)}{} \leq \robmag{\rho}{}$;
    \item \emph{Submultiplicativity under tensor product}: $\robmag{\rho^A \otimes \rho'^B}{} \leq \robmag{\rho^A }{} \robmag{\rho'^B}{}$.
\end{enumerate}
\begin{figure}[thbp]
    \centering
    \includegraphics[width=0.4\textwidth]{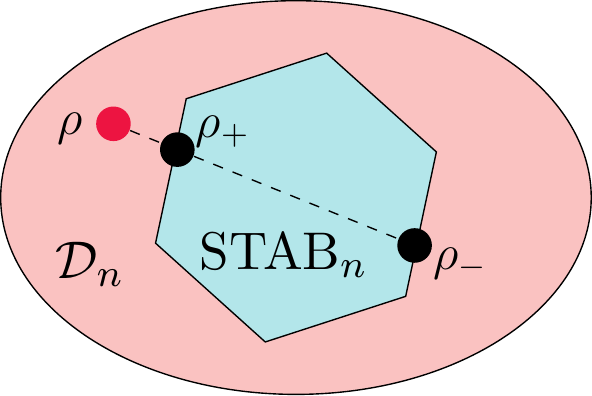}
    \caption{Schematic illustration of a density matrix $\rho \in \mathcal{D}_n$ decomposed as an affine combination of elements from the stabiliser polytope $\STABmth{n}$.}
    \label{fig:robustness_decomp}
\end{figure}
The quantity $\mathcal{R}$ also has a clear operational meaning, quantifying the classical simulation cost in a Monte Carlo-type scheme that samples from a quasiprobability distribution over stabiliser states \cite{Pashayan2015,Howard2017,Bennink2017}. These algorithms estimate the expectation value of a Pauli observable after a stabiliser channel is applied to a \emph{non}-stabiliser input state. The minimum number of samples required to achieve some stated accuracy scales with $\mathcal{R}^2$.

The robustness of magic can be calculated using standard linear programming techniques \cite{Boyd2004} (for example using the MATLAB package CVX \cite{CVX}). The naive formulation of the linear program is practical on a desktop computer for up to five qubits (the number of stabiliser states increases super-exponentially with $n$). It was recently shown by Heinrich and Gross \cite{Heinrich2018} that when states possess certain symmetries, the original optimisation problem can be mapped to a more tractable one, so that the robustness of magic can be calculated for up to 10 copies of a state. 

\begin{figure}[htbp]
    \centering
    \includegraphics[width=0.6\textwidth]{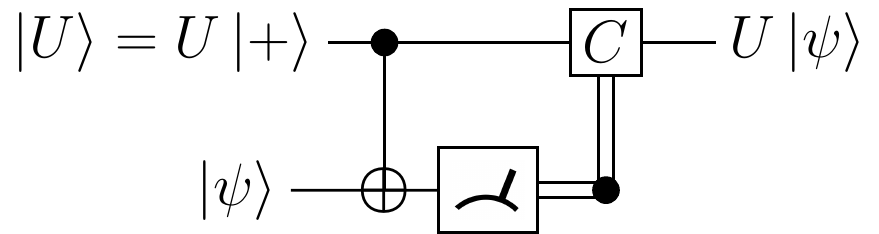}
    \caption{State injection gadget. A resource state $\ket{U}$ is consumed in order to implement the corresponding gate $U$. A Clifford correction $C$ is applied to qubit 1 conditioned on the outcome of a Pauli measurement on qubit 2. A single-qubit diagonal gate is shown, but the scheme can be generalised to all multi-qubit gates from the third level of the Clifford hierarchy.}
    \label{fig:state_injection}
\end{figure}

The framework naturally extends to a subclass of non-stabiliser circuits: those that may be implemented by deterministic state injection \cite{Howard2017}, including all gates from the third level of the Clifford hierarchy (Figure \ref{fig:state_injection}). The canonical example is the T-gate, $T = \mathrm{diag}\qty(1,e^{i \pi /4})$, which can be implemented by consuming so-called magic states as a resource \cite{Bravyi2005}. The classical simulation overhead for implementing a gate is then the robustness of magic for the consumed resource state. Not all non-stabiliser operations can be implemented in this way, however. 

Informally we say that an operation is stabiliser-preserving if it always maps stabiliser states to stabiliser states. To make this precise, define \SPO{n,m} to be the set of $n$-qubit operations $\chan$ such that  $\qty(\chan \otimes \idn{m})\sigma \in \STABmth{n+m}$ for all $\sigma \in \STABmth{n+m}$, where $\idn{m}$ is the identity map for an $m$-qubit Hilbert space. The set \SPO{n,0} is then the set of channels that map $n$-qubit stabiliser states to $n$-qubit stabiliser states.  We say a channel is ``completely'' stabiliser-preserving if $\chan \in \SPOmth{n,m}$  for all $m$.

\section{\label{sec:results}Overview of main results}

Our first result is a characterisation of the class of completely stabiliser-preserving operations, making use of the well-known Choi-Jamio{\l}kowski isomorphism~\cite{Jamiolkowski1972,Choi1975,Choi}. Throughout this paper we will consider Hilbert spaces comprised of $n$-qubit subspaces, eg. $\mathcal{H} = \mathcal{H}^A \otimes \mathcal{H}^B$. Where it is unclear from the context, we use superscripts to indicate which subspace each operator belongs to. For the maximally entangled state $\ket{\Omega_n}$, as defined in Theorem \ref{thmcor:SPOnn} below, we use the superscript $A|B$ to indicate the partition across which the state is entangled.

\newpage
\begin{theorem}[Completely stabiliser-preserving operations]\label{thmcor:SPOnn}
Given an $n$-qubit CPTP channel $\chan$, for all $m > 0 $, $\chan \in \SPOmth{n,n + m}$ if and only if $\chan \in \SPOmth{n,n}$.  Furthermore,  $\chan \in \SPOmth{n,n}$ if and only if the Choi state
\begin{equation}
	\Phi_\chan^{AB} = \outstate{\chan^A}{\op{\Omega_n}^{A|B}}{n},
	\quad \text{where} \quad
	\ket{\Omega_n}^{A|B} = \frac{1}{\sqrt{2^n}}\sum_{j=0}^{2^n - 1} \ket{j}^A \otimes \ket{j}^B,
	\label{eq:Omegastate}
\end{equation}
is a stabiliser state.  Here, $\ket{j}$ are the $n$-qubit computational basis states. 
\end{theorem}
\noindent
We prove this in section \ref{sec:SPO} .  We take this to be the set of free operations in our resource theories. 

Our first new monotone is channel robustness $\mathcal{R}_*$.   For an $n$-qubit CPTP channel $\chan$ this is defined as:
\begin{equation}
    \mathcal{R}_* (\chan) = \min_{\Lambda_{\pm}\in \SPOmth{n,n} \cap \mathrm{CPTP} }\qty{2p + 1 : (1+p)\Lambda_+ - p \Lambda_- = \chan, p \geq 0},
\end{equation}
where $\Lambda_\pm$ are completely stabiliser-preserving and CPTP  maps.  To fully enumerate this class of maps, we notice that the associated Choi state must satisfy two conditions:  (i) $\Phi_\chan^{AB}$ is a stabiliser state, and (ii) $\Phi_\chan^{AB}$ satisfies the trace-preservation condition $\Tr_A(\Phi_\chan) = \frac{\idn{n}}{2^n}$. We can therefore write:
\begin{align}
	\label{AlternateDef}
\mathcal{R}_* (\chan) = \min_{\rho_\pm \in \STABmth{2n}}\qty{2p + 1 : (1+p)\rho_+ - p \rho_- = \Phi_\chan^{AB}, p\geq0, \Tr_A(\rho_\pm) = \frac{\idn{n}}{2^n}}.
\end{align}
This can now be calculated by linear program given access to a list of all stabiliser states (see Appendix \ref{app:channel_optimisation} and the code repository Ref. \cite{channel_repo}). Channel robustness satisfies the following:
\begin{enumerate}
    \item \emph{Faithfulness}: If $\chan$ is a CPTP channel, then $\Rstar{\chan} = 1$ if $\chan$ is completely stabiliser-preserving and strictly larger than $1$ otherwise;
    \item \emph{Convexity}: $\Rstar{\sum_j q_j \chan_j} \leq \sum_j \abs{q_j} \Rstar{\chan_j}$;
    \item \emph{Submultiplicativity under composition}: $\Rstar{\chan_2 \circ \chan_1} \leq \Rstar{\chan_2} \Rstar{\chan_1}$;
    \item \emph{Submultiplicativity under tensor product}: $\Rstar{\chan^A \otimes \chan'^B} \leq \Rstar{\chan^A} \Rstar{\chan'^B}$.
\end{enumerate}
As a special case, if $\Lambda$ is a CPTP stabiliser channel, then
\begin{equation}
	\mathcal{R}_*( \Lambda  \circ \chan ) \leq \mathcal{R}_*(\chan) \mathcal{R}_*(\Lambda) = \mathcal{R}_*(\chan),
\end{equation}	
and similarly $ \mathcal{R}_*(\chan \circ \Lambda ) \leq \mathcal{R}_*(\chan)$. This combines submultiplicativity under composition and faithfulness, to show that $\mathcal{R}_*$ is suitably monotonically non-increasing under compositions with stabiliser channels. This is the sense in which channel robustness is a magic monotone for channels.  We prove submultiplicativity in Section \ref{sec:choi}. For completeness we prove convexity and faithfulness in Appendix \ref{app:channel_robustness}.  

The approach above is very close to the stabiliser decomposition of channels employed by Bennink et al. in Ref. \cite{Bennink2017}. The main difference is that Bennink et al. optimise their decomposition with respect to $\CPR$, the set of Clifford unitaries supplemented by Pauli reset channels, rather than $\SPOmth{n,n}$. The set $\CPR$ turns out to be a strict subset of the stabiliser-preserving CPTP maps, so $\mathcal{R}_*$ is a lower bound to the $\ell_1$-norm of any $\CPR$ decomposition (though the bound is tight in many cases). Just as the $\ell_1$-norm in Bennink et al. quantifies the sample complexity of a classical simulation algorithm, we can construct a related algorithm where the runtime depends on $\mathcal{R}_*$ in a similar way. We give the details of this algorithm in Section \ref{subsec:algorithm_direct}.

Before proceeding, let us reflect on the condition $\Tr_A(\rho_\pm) = \idn{n} / 2^n$ that enforces that the corresponding channels $\chan_\pm$ are trace-preserving. Dropping this condition would instead lead to $\mathcal{R}(\Phi_\chan)$, the robustness of the Choi state. For gates from the third level of the Clifford hierarchy, deterministic state injection is always possible, and hence the resource cost of the gate $U$ can be equated with the robustness of magic of the corresponding resource state.  These resource states can always (by Clifford-equivalence) be taken to have the form $\op{U} = \qty(U_A \otimes \id_B)\op{\Omega}$. This is precisely the Choi state, so it is natural to ask if $\mathcal{R}(\Phi_\chan)$ also quantifies non-stabiliserness for more general channels. We find that  $\mathcal{R}(\Phi_\chan)$ exhibits faithfulness, convexity and submultiplicativity under tensor product, but lacks submultiplicativity under composition. This arises from the fact that the decomposition of the Choi state corresponds to a decomposition of the channel into maps that are not necessarily trace-preserving. See Appendix \ref{app:choiR} for a counterexample. Despite this shortcoming, we will see that $\mathcal{R}\qty(\Phi_\chan)$ is a useful quantity to compare to more well-behaved measures. Moreover, given that $\mathcal{R}(\Phi_\chan)$ does give the resource cost for third level Clifford hierarchy gates, for consistency of the framework it is reasonable to require that our new monotones should be equal to $\mathcal{R}(\Phi_\chan)$ for this restricted class of gates. We will show later that our monotones do have this property.

Our second new monotone is the \emph{magic capacity}.  Given an $n$-qubit channel $\chan$, it is natural to consider the largest possible increase in robustness of magic, over any possible input state. By analogy with the resource theories of entanglement \cite{Campbell2010} and coherence \cite{Stahlke2014}, we define the magic capacity as:
\begin{equation*}
\mathcal{C}(\chan) = \max_{\ket{\phi} \in \STABmth{2n}} \mathcal{R}\qty[ \qty(\chan \otimes \idn{n} )\op{\phi}].
\end{equation*}
Note that the definition of capacity involves forming a tensor product of an $n$-qubit channel $\chan$ with the $n$-qubit identity. This is necessary because there exist $n$-qubit channels that generate their maximum robustness when applied to part of an $m$-qubit state, where $m>n$. Nevertheless, the $n$-qubit identity suffices for our definition; This is a consequence of Lemma \ref{lem:twontensor} in Section \ref{sec:SPO}. The capacity has the following useful properties:
\begin{enumerate}
\item \emph{Faithfulness}: If $\chan$ is a completely positive, trace-preserving (CPTP) channel, then $\mathcal{C} = 1$ if $\chan$ is stabiliser-preserving (SP), and strictly larger than $1$ otherwise;
\item \emph{Convexity}: $\Mcap{\sum_j q_j \chan_j} \leq \sum_j |q_j| \Mcap{\chan_j} $;
\item \emph{Submultiplicativity under composition}: $\Mcap{\chan_1 \circ \chan_2} \leq \Mcap{\chan_1} \Mcap{\chan_2}$;
\item \emph{Submultiplicativity under tensor product}: $\Mcap{\chan^A \otimes \chan'^B} \leq \Mcap{\chan^A} \Mcap{\chan'^B}$;
\item \emph{Maximum increase in robustness}: $\frac{\mathcal{R}[\qty(\chan \otimes \id)\rho]}{\mathcal{R}\qty(\rho)} \leq \Mcap{\chan}, \, \forall \rho $.
\end{enumerate}
In Section \ref{sec:cap} we prove properties 3-5, (convexity and faithfulness are shown in Appendix \ref{app:capacity_properties}). We will also prove the following theorem relating magic capacity to channel robustness and $\mathcal{R}\qty(\choistate{\chan})$.
\begin{theorem}[Sandwich Theorem]\label{thm:sandwich}
For any CPTP map $\chan$, the following inequalities hold:
\begin{equation}
    \RChoi{\chan} \leq \Mcap{\chan} \leq \Rstar{\chan}.
\end{equation}
Moreover, if the unitary operation $\mathcal{U}$ is in the third level of the Clifford hierarchy, then we have equality:
\begin{equation}
    \RChoi{\mathcal{U}} = \Mcap{\mathcal{U}} = \Rstar{\mathcal{U}}.
\end{equation}
\end{theorem}
We are interested in whether or not these inequalities are tight for more general operations. In Table \ref{tab:diag_results} we summarise numerical results for a selection of diagonal gates. The results for these gates are presented in full in Section \ref{sec:numerical}.

\begin{table}[htbp]
    \centering
    \begin{tabular}{c|c|c|c|c}
    
         $n$ & 2 & 3 & 4 & 5  \\
         \hline
         Multicontrol gates, $t=0$ & $\mathcal{R}_\Phi = \mathcal{C} = \mathcal{R}_*$
                                        & $\mathcal{R}_\Phi = \mathcal{C} = \mathcal{R}_*$
                                        & $\mathcal{R}_\Phi = \mathcal{C} = \mathcal{R}_*$
                                        & $\mathcal{R}_\Phi = \mathcal{C} < \mathcal{R} $ \\
         Multicontrol gates, $t\geq 1$  & $\mathcal{R}_\Phi = \mathcal{C} = \mathcal{R}_*$ 
                                        & $\mathcal{R}_\Phi = \mathcal{C} = \mathcal{R}_*$
                                        & $\mathcal{R}_\Phi = \mathcal{C} < \mathcal{R}_*$
                                        & $\mathcal{R}_\Phi < \mathcal{C} < \mathcal{R} $ \\
         Random phase gates         & $\mathcal{R}_\Phi = \mathcal{C} = \mathcal{R}_*$ 
                                        & $\mathcal{R}_\Phi = \mathcal{C} \leq \mathcal{R}_*$
                                        & $\mathcal{R}_\Phi = \mathcal{C} \leq \mathcal{R}_*$
                                        & - \\
    \end{tabular}
    \caption{Tightness of bound given by Theorem \ref{thm:sandwich}, as determined by numerical estimation of diagonal gates, where $\mathcal{R}_\Phi$ is the robustness of the Choi state, $\mathcal{C}$ is the magic capacity, $\mathcal{R}_*$ is the trace-preserving variant of $\mathcal{R}_\Phi$. Here an equality indicates that in all cases investigated, values calculated were equal up to the precision of the solver. Multicontrol phase gates are taken to be those represented by unitaries of the form $\mathrm{diag}(1,\ldots,1, \exp[i \pi/2^t])$.}
    \label{tab:diag_results}
\end{table}
The magic capacity also quantifies the sample complexity for a Monte Carlo-type classical simulation algorithm, presented in Section \ref{subsec:algorithm_convex}. This differs from previous algorithms such as Bennink et al. \cite{Bennink2017} in that a convex optimisation must be solved at each step. While this results in an increase in runtime per sample, it can be the case that $\mathcal{C}(\chan) \ll \mathcal{R}_*(\chan)$, which can lead to an improvement in sample complexity over the algorithm of Section \ref{subsec:algorithm_direct}.

\section{\label{sec:SPO}Completely stabiliser-preserving operations}
In this section, we justify setting \SPO{n,n}  as the class of free operations. We begin with an example channel $\chan \in \SPOmth{n,0}$ that fails to be stabiliser-preserving when acting on part of a larger system. Consider the single-qubit channel $\chan_T$ defined by the Kraus operators $\qty{\op{0}{T},\op{1}{T_\perp} } $, where $\ket{T} = T\ket{+}$ and $\ket{T_\perp} = T\ket{-}$. Clearly, applied to any single-qubit state, the output will be some probabilistic mixture of $\ket{0}$ and $\ket{1}$, and so must have $\mathcal{R}=1$, so $\chan_T \in \SPOmth{n,0}$. But if $\chan_T$ is applied to one qubit in a Bell pair, we obtain:
\begin{equation}
 \outstate{\chan_T}{\op{\Phi_+}}{} = \frac{1}{2} \qty(\op{0T^*} + \op{1 T_\perp^*}),   
\end{equation}
where $\ket{T^*}= T^\dagger \ket{+}$, $\ket{T^*_\perp}= T^\dagger \ket{-}$. From this output state, we can deterministically recover a pure magic state on qubit 2 using only stabiliser operations, by making a $Z$-measurement on qubit 1 and then performing a rotation on qubit 2 conditioned on the outcome. The output state has robustness $\robmag{ \outstate{\chan_T}{\op{\Phi_+}}{}}{} = \robmag{\ket{T}}{} = \sqrt{2}$. 

So, there exist channels where $\chan \in \SPOmth{n,m}$ but $\chan \notin \SPOmth{n,m + 1}$. To call a channel \emph{completely} stabiliser-preserving, then, we need to be sure $\chan \otimes \idn{m}$ remains stabiliser-preserving for all $m>0$.  We now show we only need tensor with the identity of the same dimension as the original channel.

\begin{lemma}[Maximum robustness achieved on $2n$ qubits]\label{lem:twontensor}
Let $\chan$ be an $n$-qubit quantum channel. Then for $m>0$, for any $\ket{\phi} \in \STABmth{2n+m}$, there exists some state $\ket{\psi}\in \STABmth{2n}$ such that:
\begin{equation}
{\mathcal{R}\qty[\qty(\chan^A \otimes \id_{n+m} )\op{\phi}^{AB}] = \mathcal{R}\qty[\qty(\chan^A \otimes \id_n )\op{\psi}^{AB'}]}.\end{equation}
\end{lemma}
\begin{proof}
Consider a $(2n+m)$-qubit stabiliser state $\ket{\phi}$, with partition $A|B$ between the first $n$ and last $n+m$ qubits. Ref. \cite{Fattal2004} shows that the state $\ket{\phi}^{AB}$ is local Clifford-equivalent to $p$ independent Bell pairs entangled across the partition $A|B$  (here ``local'' means with respect to the bipartition rather than per qubit). Since there are $n$ qubits in partition $A$, $p$ is at most $n$. Let $B'|B''$ be a partition of $B$ into $n$ and $m$ qubits. Then by local permutation of qubits within B, we can take these $p\leq n$ Bell pairs to be entangled across $A|B'$. So we have:
\begin{equation}
 \ket{\phi}^{AB} = \qty(\idn{n} \otimes U^B ) \ket{\psi}^{AB'} \ket{\psi'}^{B''},
\end{equation}
where $U^B$ is a Clifford operation, $\ket{\psi}^{AB'} \in \STABmth{2n}$ and $\ket{\psi'}^{B''}\in \STABmth{m}$. So writing the channel corresponding to $U^B$ as $\mathcal{U}^B$, for any $\chan$ on $n$ qubits, we know that:
\begin{align}
\mathcal{R}\qty[\qty(\chan^A \otimes \id_{n+m}) \op{\phi}^{AB}]
		&= \mathcal{R}\qty[\qty(\idn{n} \otimes \mathcal{U}^B)\qty(\qty(\chan^A \otimes \idn{n})\qty(\op{\psi}^{AB'})  \otimes    \op{\psi'}^{B''})].
\end{align}
Since $\idn{n} \otimes\mathcal{U}^B$ represents a (reversible) Clifford gate, by monotonicity of robustness of magic:
\begin{align}
\mathcal{R}\qty[\qty(\chan^A \otimes \idn{n+m}) \op{\phi}^{AB}] &= \mathcal{R}\qty[\qty(\chan^A \otimes \idn{n})\qty(\op{\psi}^{AB'})  \otimes    \op{\psi'}^{B''}] \\
& = \mathcal{R}\qty[\qty(\chan^A \otimes \idn{n})\op{\psi}^{AB'}],
\end{align}
where in the last line we used the fact that $\ket{\psi'}^{B''}$ is a stabiliser state, and hence does not contribute to the robustness. The state $\op{\psi}^{AB'}$ is a $2n$-qubit state, so this proves the result. \end{proof}

This lemma allows us to prove the first claim of Theorem \ref{thmcor:SPOnn}, which says that $\chan$  is completely stabiliser-preserving  if and only if $\chan \in \SPOmth{n,n}$. The inclusion $ \SPOmth{n,n} \subseteq  \SPOmth{n,n+m}$ is immediate since the stabiliser states are preserved under tracing out of auxiliary systems.  The interesting inclusion is $ \SPOmth{n,n+m} \subseteq  \SPOmth{n,n}$. Suppose that $\chan \in \SPOmth{n,n}$ and consider any $\sigma \in \STABmth{2n + m}$. By Lemma \ref{lem:twontensor} there exists some stabiliser state $\sigma' \in \STABmth{2n}$ such that $\robmag{\outstate{\chan}{\sigma}{n+m}}{} = \robmag{\outstate{\chan}{\sigma'}{n}}{}$. But if $\chan \in \SPOmth{n,n}$, then $\outstate{\chan}{\sigma'}{n}$ is a stabiliser state, so the robustness is equal to 1. By the faithfulness of robustness of magic, $\outstate{\chan}{\sigma}{n+m}$ is a stabiliser state. Therefore, $\chan \in \SPOmth{n,n}$ implies $\chan \in \SPOmth{n,n + m}$. Next, we discuss a straightforward test for membership of this set, which does not require mechanically checking all possible input stabiliser states.

We can associate every $n$-qubit channel $\chan$ with a unique density operator on $2n$ qubits \cite{Choi,Jamiolkowski1972,Choi1975}:
\begin{equation}
	\choistate{\chan}^{AB} = \qty(\chan^A \otimes \id^B)\op{\Omega_n}^{A|B}, \quad \text{where} \quad  \ket{\Omega_n}^{A|B} = \frac{1}{\sqrt{2^n}}\sum_{j=0}^{2^n - 1} \ket{j}^A \otimes \ket{j}^B.
	\label{eq:CJdefn}
\end{equation}
Here $\ket{j}$ label the computational basis states. We will also use the following property:
\begin{equation}
	\Tr[A \mathcal{E}\qty(\rho)] = 2^n \Tr[\Phi_\mathcal{E}(A\otimes \rho^T)], \quad \forall\,\rho,A. \label{eq:CJtrace}
\end{equation}
Consider the robustness of magic of the Choi state, $\mathcal{R}(\Phi_\chan)$. We mentioned earlier that $\mathcal{R}(\Phi_\chan)$ quantifies simulation cost for gates from the third level of the Clifford hierarchy. This motivates us to consider its properties for more general operations, and it turns out that $\RChoi{\chan}$ gives us our first criterion for completely stabiliser-preserving channels.

\begin{lemma}[Faithfulness of robustness of the Choi state]\label{thm:choi_faithfulness}
	Consider the $n$-qubit CPTP channel $\chan$. If $\chan \in \SPOmth{n,n}$, then $\RChoi{\chan} = 1$. Otherwise, $\RChoi{\chan} > 1$.
\end{lemma}

\begin{proof}
	The fact that $\chan \in \SPOmth{n,n}$ implies $\RChoi{\chan} = 1$ is easy to see. Since $\op{\Omega_n}$ is itself a $2n$-qubit stabiliser state, $\chan \in \SPOmth{n,n}$ guarantees that $\choistate{\chan}$ is a stabiliser state, so must have robustness 1. The implication in the other direction is less obvious; one might imagine there perhaps exist maps that send $\op{\Omega_n}$ in particular to a stabiliser state, but are not stabiliser-preserving in general. We show that this is not the case using an argument based on witnesses for non-stabiliser states, in part inspired by the conditions for free operations (SPO) given by Ahmadi et al. \cite{Ahmadi} for odd prime dimension qudits and for the single qubit case. The criteria for SPO were based on a class of witness defined by phase point operators. Here we instead consider the following family of witnesses for $n$ qubits. We say that $W_n$ is a good witness for $n$-qubit non-stabiliser states if:
	\begin{equation}
		\Tr(W_n \sigma ) \leq 0, \quad \forall \sigma \in \STABmth{n}.
		\label{eq:goodwitness}
	\end{equation}
	The hyperplane separation theorem \cite{Boyd2004} guarantees that such witnesses exist and can be constructed for any non-stabiliser state $\rho$. That is, for any $\rho \notin \STABmth{n}$, there always exists an operator $W_\rho$ such that $\Tr(W_\rho \rho)>0$ and yet is a good $n$-qubit witness as defined above.
	
	We first show that for any good $n$-qubit witness $W_n$, the operator $W_n \otimes \op{\phi}$, where $\ket{\phi}\in \STABmth{m}$, is a good witness for $(n + m)$-qubit non-stabiliser states. For any $\sigma \in \STABmth{n+m}$:
	\begin{align}
		\Tr[( W_n \otimes \op{\phi}) \sigma] & = \Tr[ \qty(\idn{n}\otimes \op{\phi}) \qty(W_n \otimes \idn{m}) \sigma] = \Tr[\qty(W_n \otimes \idn{m}) \widetilde{\sigma} ],
	\end{align}
	so that ${\widetilde{\sigma} = \qty(\idn{n}\otimes \op{\phi}) \sigma \qty(\idn{n}\otimes \op{\phi}) }$, where we used cyclicity of the trace and the fact that $\op{\phi}$ is a projector. If $\widetilde{\sigma}=0$ then the inequality \eqref{eq:goodwitness} is trivially always satisfied by $W_n \otimes \op{\phi}$.  Otherwise, $\widetilde{\sigma}$ is a stabiliser state (non-normalised) and so too is $\Tr_B[\widetilde{\sigma}]$.   Then:
		\begin{align}
			\Tr[\qty(W_n \otimes \idn{m}) \widetilde{\sigma} ] & =  \Tr[W_n \Tr_B\qty(\widetilde{\sigma}) ] \leq 0.
		\end{align}
		The inequality follows because $W_n$ is a good witness and $\Tr_B[\widetilde{\sigma}]$ is a stabiliser state. Therefore, $W_n \otimes \op{\phi}$ is also a valid witness.
	
	Now suppose $\chan \notin \SPOmth{n,n}$. Then there is some stabiliser state $\ket{\phi'} \in \STABmth{2n}$, such that $\rho' = \qty(\chan \otimes \idn{n}) \op{\phi'} \notin \STABmth{2n}$. By the hyperplane separation theorem, there exists a good $2n$-qubit witness $W_{\rho'}$ such that $\Tr[W_{\rho'} \rho']>0$. Consider that the $4n$-qubit state $\ket{\Omega_{2n}}^{AA'|BB'}$ is unentangled across the partition $AB|A'B'$, so that we can write:
	\begin{equation}
	\ket{\Omega_{2n}}^{AA'|BB'} = \frac{1}{2^{2n}} \sum_{j,k} \ket{j}^{A} \ket{k}^{A'} \otimes \ket{j}^{B} \ket{k}^{B'}= \ket{\Omega_n}^{A|B}\otimes \ket{\Omega_n}^{A'|B'},
	\end{equation}
	taking care to note the permutation of subspaces. Therefore the Choi state for $\qty(\chan \otimes \idn{n})$ is:
	\begin{equation}
	    \Phi^{AA'|BB'}_{\chan \otimes \idn{n}} =(\chan^A \otimes \idn{n}^{A'}) \otimes \idn{2n}^{BB'}(\op{\Omega_{2n}}^{AA'|BB'}) = \Phi_\chan^{AB} \otimes \op{\Omega_n}^{A'|B'}.
	\end{equation}
     We then use equation \eqref{eq:CJtrace} to obtain:
	\begin{align}
		0 < \frac{1}{2^n}\Tr[W_{\rho'} \rho'] & = \Tr[\Phi_{\chan \otimes \idn{n}} (W_{\rho'} \otimes \op{\phi'}^T)], \quad \text{where} \quad \op{\phi'}^T \in \STABmth{2n}. 
	\end{align}
	But $(W_{\rho'} \otimes \op{\phi'}^T)$ is a good witness, so $\Phi_\chan \otimes \op{\Omega_n} \notin \STABmth{4n}$ and therefore $\Phi_\chan$ is a non-stabiliser state. So, by faithfulness of robustness of magic, if $\chan \notin \SPOmth{n,n}$ then $\mathcal{R}\qty(\Phi_\chan)> 1$.
\end{proof}
Combined the above two lemmas provide a proof of both claims given in Theorem \ref{thmcor:SPOnn}.  This does not mean the robustness of the Choi state is a reliable monotone, since despite being faithful it fails to be submultiplicative under composition (see Appendix \ref{app:choiR}).  Rather, we use the faithfulness of the Choi state as a tool to give an alternative definition of the channel robustness as captured by Eq.~\eqref{AlternateDef}.

\section{\label{sec:choi}Channel robustness}

A natural extension of the algorithm of Bennink et al.  \cite{Bennink2017} is to replace $\CPR$ (the set of Clifford gates and Pauli reset channels) with \SPO{n,n}. We therefore define the \emph{channel robustness} as:
\begin{equation}
    \mathcal{R}_* (\chan) = \min_{\Lambda_\pm \in \SPOmth{n,n}}\qty{2p + 1 : (1+p)\Lambda_+ - p \Lambda_- = \chan, p \geq 0}.
\end{equation}
To calculate this in practice, we decompose the Choi state $\Phi_\chan$ as per equation \eqref{AlternateDef}, adapting the robustness of magic optimisation problem from Ref. \cite{Howard2017}. The details are given in Appendix \ref{app:channel_optimisation}. We also note that for diagonal channels $\chan$, the problem is equivalent to a decomposition of the state $\chan(\op{+}^{\otimes n})$ as:

\begin{equation}
    \chan(\op{+}^{\otimes n}) = (1+p) \rho_+ - p \rho_-, \, \text{where} \,\, \rho_{\pm} \in \STABmth{n}, \, \, \text{and} \bra{x}\rho_\pm \ket{x}  = \frac{1}{2^n}, \,\, \forall{x}. \label{eq:diagonal_decomp}
\end{equation}
Here, the condition on the partial trace for the general case is replaced by the requirement that all diagonal elements of the states $\rho_\pm$ are equal to $1/2^n$. In the special case where $\rho_\pm$ are pure, this implies they are diagonal Clifford-equivalent to graph states, though more generally they may be mixtures of states that individually do not have full support in the standard basis. In practical terms, this reduction to an $n$-qubit problem confers a significant advantage, since the number of stabiliser states (which form the extreme points of the linear programming problem) grows super-exponentially with $n$. Full technical details of this simplification are given in Appendix \ref{app:diagonal}. We now return to consider the properties of channel robustness.

The channel robustness is convex and faithful with these properties inherited from the robustness of magic (see Appendix \ref{app:channel_robustness} for details).  Here we discuss additional properties.

\textbf{Submultiplicativity under composition:} $\mathcal{R}_*(\chan_2 \circ \chan_1 ) \leq \mathcal{R}_*(\chan_1) \mathcal{R}_*(\chan_2)$.  The channels $\chan_1$ and $\chan_2$ will have an optimal decomposition:
\begin{equation}
	\chan_j = (1 + p_j) \Lambda_{j,+} - p_j \Lambda_{j,-} ,
\end{equation}
where $\mathcal{R}_*(\chan_j)=1+2p_j$ and $\Lambda_{j,\pm}$ are CPTP maps and completely stabiliser preserving.  Using these decompositions, we obtain that
\begin{equation}
	\chan_2 \circ \chan_1  = (1+q)  \Lambda_{+}'  - q \Lambda_{-}'  ,
\end{equation}
where 
\begin{align}
	\Lambda_{+}' & = (1+q)^{-1} [  (1 + p_2)(1 + p_1) \Lambda_{2,+}  \circ \Lambda_{1,+} + p_2p_1 \Lambda_{2,-}  \circ \Lambda_{1,-} ]   , \\
	\Lambda_{-}' & =q^{-1} [ p_2(1 + p_1) \Lambda_{2,-}  \circ \Lambda_{1,+} +   (1+p_2)p_1 \Lambda_{2,+}  \circ \Lambda_{1,-}  ] , \\
	q & = p_1 + p_2 + 2 p_1 p_2
\end{align}
The set of CPTP completely stabiliser preserving channels is closed under composition and convex, so both  $\Lambda_{\pm}'$ are in this set.   Therefore, we have a valid decomposition for $\chan_2 \circ \chan_1$ that entails $\mathcal{R}_*(\chan_2 \circ \chan_1 )\leq 1+2q$.  One finds  
\begin{align}
	1+ 2 q & = (1+2p_1)(1+2p_2) = \mathcal{R}_*(\chan_1) \mathcal{R}_*(\chan_2) ,
\end{align}
which completes the proof.  

\textbf{Submultiplicativity under tensor product:} $\Rstar{\chan^A \otimes \chan'^B} \leq \Rstar{\chan^A} \Rstar{\chan'^B}$. We treat tensor product as a special case of composition. For $n$-qubit $\chan_A$ and $m$-qubit $\chan'_B$:
\begin{equation}
\mathcal{R}_* (\chan^A \otimes \chan'^B) \leq \mathcal{R}_* (\chan^A \otimes \idn{m}^B) \mathcal{R}_* (\idn{n}^A \otimes \chan'^B) .
\end{equation}
To complete the proof we will confirm that 
\begin{equation}
\mathcal{R_*} (\idn{n}^A \otimes \chan^B) = \mathcal{R_*} (\chan^A \otimes \idn{n}^B) = \mathcal{R_*} (\chan).
\label{eq:RidentityEquality}
\end{equation}
As noted earlier, we can write 
$
\ket{\Omega_{n+m}}^{AA'|BB'} = \ket{\Omega_{n}}^{A|B} \otimes \ket{\Omega_{m}}^{A'|B'}
$, so that the Choi state for $\chan^A \otimes \idn{m}$ is given by:
\begin{align}
\Phi_{\chan \otimes \idn{m}}^{AA'|BB'} &= \qty(\chan^A \otimes \idn{m}^{A'} \otimes {\idn{n+m}^{BB'}} ) \op{\Omega_{n+m}}^{AA'|BB'} \nonumber \\
		& = \qty(\chan^A \otimes \idn{n}^B) \op{\Omega_{n}}^{A|B} \otimes  \op{\Omega_{m}}^{A'|B'}\nonumber  \\
        & = \Phi_{\chan}^{AB} \otimes \op{\Omega_{m}}^{A'|B'}.
\end{align}
The state $\Phi_{\chan}$ will have some optimal decomposition $\Phi_{\chan}  = {(1+p)\rho_+ - p \rho_-}$, with channel robustness $\mathcal{R}_*(\chan_A) = 1 +2p$, so that:
\begin{align}
\Phi_{\chan \otimes \idn{m}}^{AA'|BB'}  &=  (1 + p) \rho_+^{AB} \otimes \op{\Omega_{m}}^{A'|B'}  - p \rho_-^{AB} \otimes \op{\Omega_{m}}^{A'|B'}.
 \label{eq:productDecomp}
\end{align}
This is a valid, not necessarily optimal, stabiliser decomposition satisfying the trace condition, so we have
\begin{equation}
\mathcal{R}_* (\chan^A \otimes \id^{A'}) \leq \mathcal{R}_* (\chan^A).
\label{eq:RidentityInequality}
\end{equation}
This is enough to show submultiplicativity; for completeness, in Appendix \ref{app:channel_robustness} we will also show ${\mathcal{R}_*(\chan) \leq \mathcal{R}_*(\chan \otimes \id)}$ so that in fact we have equality.

\section{\label{sec:cap}Magic capacity}
\subsection{Properties}
We now turn to our second monotone, which quantifies the capacity of a channel to generate magic. Recall:
\begin{equation}
\mathcal{C}(\chan) = \max_{\ket{\phi} \in \STABmth{2n}} \mathcal{R}\qty[ \qty(\chan \otimes \idn{n} )\op{\phi}] \label{eq:capdef},
\end{equation}
where $\mathcal{R}$ is the robustness of magic. Notice that we only need optimise over the pure stabiliser states.  For mixed states or even non-stabiliser states, the capacity still captures the possible increase in robustness of magic by virtue of the \textbf{maximum increase in robustness} property:
\begin{equation}
	\frac{\robmag{\outstate{\chan}{\rho}{n}}{}}{\robmag{\rho}{}} \leq \Mcap{\chan}
	\label{eq:maxincrease}.
\end{equation}
Here we prove this property, using similar arguments to those deployed in \cite{Campbell2010}.  Consider an $n$-qubit channel $\chan$. Any $2n$-qubit input state $\rho$ will have an optimal stabiliser state decomposition $\rho = \sum_j q_j \op{\phi_j}$, where $\sum_j q_j = 1$, and such that 
	$\robmag{\rho}{} = \sum_j |q_j|$. By linearity we have:
	\begin{equation}
		(\chan \otimes \idn{n})\rho = \sum_j q_j (\chan \otimes \idn{n}) \op{\phi_j}.
	\end{equation}
	By convexity of robustness of magic, we then have:
	\begin{align}
		\robmag{\outstate{\chan}{\rho}{n}}{} &\leq \sum_j |q_j|\, \robmag{(\chan \otimes \idn{n}) \op{\phi_j}}{}.
	\end{align}
	The optimal pure stabiliser state $\ket{\phi_*}$, satisfies: $$\Mcap{\chan} = \robmag{\outstate{\chan}{\op{\phi_*}}{n}}{} \geq \robmag{\outstate{\chan}{\op{\phi_j}}{n}}{}$$ for any $j$.  So we have:
	\begin{align}
		\robmag{\outstate{\chan}{\rho}{n}}{} &\leq \robmag{\outstate{\chan}{\op{\phi_*}}{n}}{} \sum_j |q_j|  = \Mcap{\chan} \robmag{\rho}{}.
	\end{align}
	Rearranging we obtain inequality \eqref{eq:maxincrease}.

\textbf{Submultiplicativity under composition:} $\Mcap{\chan_1 \circ \chan_2} \leq \Mcap{\chan_1} \Mcap{\chan_2}$.  Take the composition of two linear maps $\chan_1$ and $\chan_2$. There exists some stabiliser state $\rho_* = \op{\phi_*}$ that achieves the optimal robustness:
\begin{align}\Mcap{\chan_2 \circ \chan_1} &= \robmag{\qty[(\chan_2 \circ \chan_1) \otimes \idn{n}] \rho_* }{}  = \robmag{(\chan_2 \otimes \idn{n}) \circ (\chan_1 \otimes \idn{n}) \rho_* }{}.
\end{align}
The operator $(\chan_1 \otimes \idn{n}) \rho_* $ will have some  optimal decomposition $(\chan_1 \otimes \idn{n}) \rho_* = \sum_k q_{1k} \op{\phi_k}$ such that $\robmag{(\chan_1 \otimes \idn{n}) \rho_*}{} = \sum_k \abs{q_{1k}}$. So by linearity:
\begin{align}
(\chan_2 \otimes \idn{n}) \circ (\chan_1 \otimes \idn{n}) \qty[\rho_*] & = (\chan_2 \otimes \idn{n}) \qty[ \sum_k q_{1k} \op{\phi_k}  ]  = \sum_k q_{1k} (\chan_2 \otimes \idn{n})\op{\phi_k}.
\end{align}
Then by convexity of robustness of magic:
\begin{align}
 \robmag{(\chan_2 \otimes \idn{n}) \circ (\chan_1 \otimes \idn{n}) \qty[\rho_*] }{} & = \robmag{\sum_k q_{1k} (\chan_2 \otimes \idn{n})\op{\phi_k}}{}\\
 & \leq \sum_k \abs{q_{1k}} \robmag{(\chan_2 \otimes \idn{n})\op{\phi_k}}{} \label{eq:subMultStep1}\\
 & \leq \sum_k \abs{q_{1k}} \Mcap{\chan_2} \label{eq:subMultStep2}\\
 & = \robmag{(\chan_1 \otimes \idn{n})\rho_*}{}  \Mcap{\chan_2},
\end{align}
where to go from \eqref{eq:subMultStep1} to \eqref{eq:subMultStep2} we used the fact that since $\ket{\phi_k}$ are stabiliser states,  $\robmag{(\chan_2 \otimes \id)\op{\phi_k}}{}$ can be no larger than $\Mcap{\chan_2}$. Finally, using the fact that $\robmag{(\chan_1 \otimes \id) \rho_*}{} \leq \Mcap{\chan_1}$, we have
$\Mcap{\chan_2 \circ \chan_1} \leq \Mcap{\chan_2} \Mcap{\chan_1}$, completing the proof.

\textbf{Submultiplicativity under tensor product}: $\Mcap{\chan^A \otimes \chan'^B} \leq \Mcap{\chan^A} \Mcap{\chan'^B}$. This follows directly from submultiplicativity under composition, since
\begin{align}
\Mcap{\chan^A \otimes \chan'^B} & =  \Mcap{\qty(\chan^A \otimes \idn{m}^B) \circ \qty(\idn{n}^A \otimes \chan'^B)  }  \leq \Mcap{\chan^A \otimes \idn{m}^B} \Mcap{\idn{n}^A \otimes \chan'^B}.
\label{eq:submultCapTensor}
\end{align}
We saw in Section \ref{sec:SPO} that any gains in robustness achievable by tensoring $\chan_A$ with the identity and acting on a larger Hilbert space are already taken care of by the $\otimes \idn{n}$ in the definition \eqref{eq:capdef}, so that  $\Mcap{\chan^A \otimes \idn{m}^B} = \Mcap{\chan^A}$ and $\Mcap{\idn{n}^A \otimes \chan'^B} = \Mcap{\chan'^B}$. Substituting this into inequality \eqref{eq:submultCapTensor} gives the desired result.

\subsection{Sandwich theorem\label{subsec:sandwich}}
We will now prove Theorem \ref{thm:sandwich}, which stated that $\RChoi{\chan} \leq \Mcap{\chan} \leq \Rstar{\chan}$, for any CPTP channel $\chan$, and that $\RChoi{\mathcal{U}} = \Mcap{\mathcal{U}} = \Rstar{\mathcal{U}}$ for any unitary operation $\mathcal{U}$ from the third level of the Clifford hierarchy.

\begin{proof}
By definition $\choistate{\chan} = \outstate{\chan}{\op{\Omega_n}}{n}$. But $\ket{\Omega_n}$ is a stabiliser state, so $\RChoi{\chan}$ can be no larger than $\robmag{\outstate{\chan}{\op{\phi_*}}{n}}{} = \mathcal{C}(\chan)$, where $\ket{\phi_*}$ is the stabiliser state that achieves the capacity, and so:
\begin{equation}
    \mathcal{R}\qty(\choistate{\chan}) \leq \mathcal{C}(\chan).
\end{equation}

Now suppose $\chan = (1+p) \Lambda_+ - p \Lambda_-$ is the optimal decomposition of $\chan$ into CPTP stabiliser-preserving maps, $\Lambda_\pm \in \SPOmth{n,n}$, so that $\Rstar{\chan} = 1 + 2p$. Then for any \emph{input} stabiliser state $\sigma \in \STABmth{2n}$, we can write down a valid stabiliser decomposition of the \emph{output} state: 
\begin{equation}
    \outstate{\chan}{\sigma}{n} = (1 + p) \outstate{\Lambda_+}{\sigma}{n} - p \outstate{\Lambda_-}{\sigma}{n}.
    \label{eq:nonoptimal}
\end{equation}
In particular this is true for the stabiliser state $\sigma_* = \op{\phi_*}$ that is optimal with respect to the capacity. But equation \eqref{eq:nonoptimal} could be a non-optimal decomposition, so its $\ell_1$-norm $1+2p$ is at least as large as $\robmag{\outstate{\chan}{\sigma_*}{n}}{}$. So:
\begin{equation}
    \mathcal{C}(\chan) = \robmag{\outstate{\chan}{\sigma_*}{n}}{} \leq 1 + 2p = \Rstar{\chan},
\end{equation}
completing the proof of the first statement. Having done so, to prove the second statement it suffices to show that $\RChoi{\mathcal{U}} = \Rstar{\mathcal{U}}$.

For any $n$-qubit gate $U$ from the third level of the Clifford hierarchy, corresponding to the channel $\mathcal{U}$, deterministic state injection is possible \cite{Bravyi2005,Gottesman1999}. That is, given a Hilbert space $\Hilb = \Hilb_A \otimes \Hilb_B \otimes \Hilb_C \otimes \Hilb_D$, where each subspace is comprised of $n$ qubits, there exists a completely-stabiliser-preserving circuit $\Lambda$ such that, for any $2n$-qubit input state $\rho$:
\begin{equation}
    \Tr_{BC}\qty[\Lambda(\Phi_\mathcal{U}^{AB} \otimes \rho^{CD})] =  \mathcal{U} \otimes \idn{n} (\rho^{AD})
    \label{eq:injectchannel}
\end{equation}
Where $\Phi_\mathcal{U}$ is the Choi state for the channel $\mathcal{U}$. The circuit $\Lambda$ is comprised of a complete Bell measurement on $BC$, followed by a Clifford correction on subspace $A$ conditioned on the outcome of the Bell measurement. It can be represented by Kraus operators:
\begin{equation}
    K_j = (C_j^A \otimes \idn{3n}^{BCD})M_j ,
\end{equation}

where $M_j = \id^A \otimes \op{\Phi_j}^{BC} \otimes \id^D$ are the Kraus operators corresponding to elements of the Bell basis $\ket{\Phi_j}$, and $C_j$ is some unitary Clifford correction.

Now consider an optimal decomposition of the Choi state:
\begin{equation}
    \Phi_{\mathcal{U}} = (1+p) \rho_+ - p \rho_+, \quad \text{s.t.} \quad \RChoi{\mathcal{U}} = 1 + 2p.
\end{equation}
We now show that by substitution into equation \eqref{eq:injectchannel} we can obtain a decomposition of the channel that satisfies the trace-preservation condition required for channel robustness. We have:
\begin{align}
    \Phi_\mathcal{U}^{AD} & = \mathcal{U} \otimes \idn{n} (\op{\Omega}^{AD}) = \Tr_{BC}\qty[\Lambda(\Phi_\mathcal{U}^{AB} \otimes \op{\Omega}^{CD})] \\
    & = (1+p) \widetilde{\rho}_+^{AD} - p \widetilde{\rho}_-^{AD}, \label{eq:thirdLchanneldecomp}
\end{align}
where $\widetilde{\rho}_\pm^{AD} = \Tr_{BC}[\Lambda(\rho_\pm^{AB} \otimes \op{\Omega}^{CD})] \in \STABmth{2n}$, since $\Lambda \in \SPOmth{n,n}$. If we can show that $\Tr_A(\widetilde{\rho}_\pm) = \idn{n}/2^n$, then we have satisfied the required condition. First, note that $\Tr_A(\widetilde{\rho}_\pm)$ is independent of the Clifford corrections $C_j$, since the partial trace depends only on the outcome probabilities of the Bell measurement, $p_j = \Tr [ M_j (\rho_\pm^{AB} \otimes \op{\Omega}^{CD}) M_j^\dagger]$. Therefore $ \Tr_A(\widetilde{\rho}_\pm) = \Tr_{ABC}(\rho'_\pm)$, where $\rho'_\pm = \sum_j M_j  (\rho_\pm^{AB} \otimes \op{\Omega}^{CD}) M_j^\dagger$ is the state following the Bell measurement. We then have:
\begin{align}
    \Tr_A(\widetilde{\rho}_\pm) &= \sum_j \Tr_{ABC}\qty[ M_j \qty(\rho_\pm^{AB} \otimes \op{\Omega}^{CD}) M_j^\dagger ] \\ 
                        & = \Tr_{ABC}\qty[\sum_j M_j^\dagger M_j  \qty(\rho_\pm^{AB} \otimes \op{\Omega}^{CD})] \\
                        & = \Tr_{ABC} \qty[\rho_\pm^{AB} \otimes \op{\Omega}^{CD}] \\
                        & = \Tr_{C} \qty[\op{\Omega}^{CD}] = \frac{\idn{n}}{2^n}.
\end{align}
In going to the second line, we used the fact that the partial trace over $BC$ is cyclic with respect to operators that act non-trivially only on $\Hilb_B \otimes \Hilb_C$. In going from the second to the third line, we used the fact that $\{ M_j \}$ is a complete set of Kraus operators, so $\sum_j M_j^\dagger M_j = \id$. We have shown that the decomposition \eqref{eq:thirdLchanneldecomp} satisfies the trace preservation criterion. Since the decomposition may not be optimal, we have that $\Rstar{\mathcal{U}}\leq 1 + 2p = \RChoi{\mathcal{U}}$. But from the proof of the first statement $\RChoi{U} \leq \Mcap{\mathcal{U}}\leq \Rstar{\mathcal{U}}$, so it must be that equality holds.
\end{proof}
We note that the result that $\mathcal{R}(\Phi_\mathcal{U}) = \Rstar{\mathcal{U}}$ for third-level gates carries over to the case of decompositions of $\mathcal{U}(\op{+}^{\otimes n })$ for diagonal third-level gates. That is, there always exists a decomposition satisfying the constraints of equation \eqref{eq:diagonal_decomp} that is optimal with respect to $\mathcal{R}(\Phi_\mathcal{U}) = \mathcal{R}(\mathcal{U}(\op{+}^{\otimes n }))$. This can be seen by following the argument of Theorem \ref{thm:sandwich}, but replacing the full $4n$-qubit teleportation circuit with a $2n$-qubit state injection circuit (Figure \ref{fig:state_injection}).

\section{\label{sec:algorithms}Classical simulation algorithms}

Here we propose two classical simulation algorithms.  The channel robustness $\mathcal{R}_*$  relates to the runtime of our first simulator, which we call the \emph{static} simulator.  The magic capacity $\mathcal{C}$ relates to the runtime of the second simulator, called the \emph{dynamic} simulator.  In both cases, we consider a circuit composed from a sequence of channels  with $\{  \mathcal{E}_1,  \mathcal{E}_2, \ldots  , \mathcal{E}_L \} $ acting on an initial stabiliser state, which we take to be $\ket{0^n}$.  The circuit ends with some final state $\rho =\mathcal{E}_L  \ldots \circ \mathcal{E}_2 \circ \mathcal{E}_1 ( \ket{0^n}\bra{0^n} )$ and measurement of some Pauli observable $Z$.  We assume that each channel acts non-trivially on a bounded number of qubits (e.g. 2 or 3) so we can evaluate the relevant monotones. Our goal is to estimate the expectation value $\mathrm{Tr}[Z \rho]$ to within additive error. In the language of Ref.~\cite{pashayan2017estimation} our simulators will be poly-boxes.  

Both our algorithms are inspired by previous methods that collect a large number of Monte Carlo samples that scales quadratically with the negativity of some quasiprobability distribution~\cite{Pashayan2015,Bennink2017,Howard2017}.  The static Monte Carlo simulator uses a precomputed, and therefore static, quasiprobability distribution.  The dynamic Monte Carlo simulator recomputes optimal quasiprobability distributions at each step, which can lead to fewer samples required but with a higher runtime per sample.  As such, there are subtle trade-offs in the runtime complexities. 

Both our algorithms use that completely stabiliser preserving operations \SPO{n,n} acting on a stabiliser state can be classically efficiently simulated.  This follows from the fact that given the Choi state $\Phi_{\mathcal{E}}$ for an $n$-qubit channel,  the channel may be implemented by performing a Bell measurement on $\Phi_{\mathcal{E}} \otimes  \sigma$, postselecting on the $\Omega$ outcome to obtain 
$ (\id_n \otimes \op{\Omega} )\Phi_{\mathcal{E}} \otimes  \sigma (\id_n \otimes \op{\Omega} )$
and then tracing out the last $2n$ qubits \cite{Gottesman1999}. This can be simulated using Gottesman-Knill when $\sigma$ is a stabiliser state and $\chan \in \SPOmth{n,n}$. Curiously, it is unclear whether \SPO{n,n} can be physically realised using Clifford unitaries and Pauli measurements but without the use of postselection.

\subsection{Static Monte Carlo\label{subsec:algorithm_direct}}

In Algorithm \ref{fig:static} we give pseudocode for our first simulator, which we call static Monte Carlo, and which can be viewed as a generalisation of the algorithm of Bennink et al. \cite{Bennink2017}, differing in two important ways. First, whereas their algorithm involved an optimisation over the set CPR, the channel robustness is optimised with respect to \SPO{m,m}, so that all completely stabiliser-preserving maps are represented non-negatively. Second, while the positive and negative parts of the decomposition are each CPTP, the subroutine in step (c) of the algorithm involves updating the state with maps that may not be trace-preserving. This is necessary to ensure that the updated state after step (c)iii. remains pure. We explain below how this subroutine is carried out and show that it does not increase the sample complexity.

We assume a pre-computation stage in which for each circuit element we determine an optimal decomposition as per the definition \eqref{AlternateDef}, so that:
\begin{equation}
\mathcal{E}_j = (1+p_{j}) \mathcal{E}_{j,0} - p_{j}  \mathcal{E}_{j,1}, \quad \text{where} \quad \Rstar{\chan_j} = 1 + 2p_j
\label{eq:circuitElPair}
\end{equation}
where  $\mathcal{E}_{j,k} \in \SPOmth{m,m}$. The runtime of this pre-computation stage is bounded since each circuit element acts non-trivially on only $m$ qubits. Once these decompositions have been calculated, one can define an overall quasiprobability distribution:
\begin{equation}
	\label{Product_q_distro}
	 q_{\vec{k}}= \prod_{j: k_j=0}  \left(  1+p_{j} \right)  \prod_{j: k_j=1} \left( -p_{j} \right)  ,
\end{equation}
where $\vec{k} \in \mathbb{Z}_2^L$ is a vector representing a choice of either $\mathcal{E}_{j,0}$ or $\mathcal{E}_{j,1}$ at each circuit element $\mathcal{E}_{j}$. We can renormalise this to obtain a product probability distribution:
\begin{equation}
	\label{Product_distro}
	 p_{\vec{k}}= \prod_{j: k_j=0}  \frac{\left(  1+p_{j} \right)}{\Rstar{\chan_j}}  \prod_{j: k_j=1}  \frac{\left( p_{j} \right)}{\Rstar{\chan_j}}.
\end{equation}
Hence we can write the output of the circuit as follows:
\begin{equation}
\chan(\op{\phi_0}) = R \sum_{\vec{k}} p_{\vec{k} } \lambda_{\vec{k}}  \Lambda_{\vec{k}}(\op{\phi_0})
\end{equation}
where each $\Lambda_{\vec{k}} = \chan_{L,k_L} \circ \ldots \circ \chan_{1,k_1}$ gives a trajectory of \SPO{m,m} channels through the circuit, $\lambda_{\vec{k}} = \mathrm{sign}(q_{\vec{k}})$, and $R = \prod_{j=1}^L \Rstar{\chan_j}$. 
The expectation value of an observable $Z$ at the end of the circuit is given by:
\begin{equation}
    \Tr[Z \chan(\op{\phi_0})] = R \sum_{\vec{k}} p_{\vec{k}} \lambda_{\vec{k}}  \Tr[Z  \Lambda_{\vec{k}}(\op{\phi_0})]
\end{equation}
This decomposition is a quasi-probability distribution with $\ell_1$-norm
\begin{equation}
\label{Runtime}
	|| \vec{q} ||_1=	\sum_{\vec{k}} |  p_{\vec{k} } \lambda_{\vec{k}} R| = \sum_{\vec{k}}|  p_{\vec{k} }| R =  \prod_{j}\mathcal{R}_*(	\mathcal{E}_j) .
\end{equation}
By sampling from the stabiliser-preserving trajectories $\Lambda_{\vec{k}}$ with probability distribution $\{p_{\vec{k}}\}$, and calculating $ \lambda_{\vec{k}} R \Tr[Z\Lambda_{\vec{k}}(\op{\phi_0})]$, we have an unbiased estimator for $\Tr[Z \chan(\op{\phi_0})]$. Crucially, \eqref{Product_distro} gives a product distribution, so this first sampling step is efficient. The variance is increased by $R>1$, but by standard arguments \cite{Pashayan2015,Bennink2017,Howard2017} the Hoeffding inequalities can be used to show that any constant error $\delta > 0$ in the mean estimate can be achieved with arbitrary success probability $(1-\epsilon)$ by repeating the sampling procedure $N$ times, where:
\begin{equation}
    N = \frac{2}{\delta}|| \vec{q} ||^2_1 \ln(\frac{2}{\epsilon}).
\end{equation}

\begin{algo}[ht] 
			\raggedright
			{\bf Input:} 
			Circuit description $\qty{\chan_1,\chan_2,\ldots,\chan_L}$, where each $\chan_j$ has a known optimal decomposition with channel robustness $\Rstar{\chan_j}$ as per equations \eqref{eq:circuitElPair} and \eqref{eq:LPoutput}, input stabiliser state $\ket{\phi_{0}}$, Pauli observable $Z$, number of samples $M$.\\
			{\bf Output:} Estimate of the Pauli expectation value $\expval{Z} = \Tr[Z \chan(\op{\phi_{0}})]$
\begin{enumerate}
	\item Set $i \leftarrow 1$, $T \leftarrow 0$ and $R \leftarrow \prod_{j=1}^L \Rstar{\chan_j}$;
		\item For $i = 1$ to $M$: 
		\begin{enumerate}	
	        \item Sample vector $\vec{k}_i$ according to the distribution $\qty{p_{\vec{k}}}$; 
	        \item Set the input state of the circuit to be the stabiliser state $\ket{\phi_{0}}$;
	        \item For $j = 1$ to $L$:
	        \begin{enumerate}
	            \item Calculate the distribution: \begin{equation*} \qty{p_l^{(j)} = \Tr[\curlyT_{l,j,k_j}(\op{\phi_{j-1}})] }; \end{equation*}
	            \item Sample $l$ with probability $ p_l^{(j)}$;
	            \item Set $\op{\phi_{j}} = \curlyT_{l,j,k_j}(\op{\phi_{j-1}})/ p_l^{(j)}$.
	        \end{enumerate}

	        \item Calculate $\langle P \rangle_i = \Tr[P\op{\phi_{L}}]$;
	        \item Set $T_{i} = \mathrm{sign}\qty(q_{\vec{k}_i}) R \langle P \rangle_i$;
	        \item $T \leftarrow T + T_i$.
     \end{enumerate}	
\end{enumerate}
			{\bf Return:} $T/M$
	\caption{\raggedright A classical simulator with sample complexity $\prod_{j=1}^L\mathcal{R}_*(\chan_j)^2$.}\label{fig:static}
\end{algo}

The simulator proceeds by tracking the evolution of a pure stabiliser state through the sampled trajectory $\Lambda_{\vec{k}}$, as described in Algorithm \ref{fig:static}. In practice, while each $\chan_{j,k_j}$ is constrained to be CPTP, the output of the linear program will be a decomposition of $\chan_{j,k_j}$ into maps that are stabiliser-preserving but not necessarily trace-preserving, corresponding to Choi states that are proportional to pure stabiliser projectors:
\begin{equation}
    \chan_{j,k_j} = \sum_l \curlyT_{l,j,k_j} \quad \longleftrightarrow \quad
    \Phi_{\chan_{j,k_j}} = \sum_l |q_{l,j,k_j}|\op{\psi_{l,j,k_j}}.
    \label{eq:LPoutput}
\end{equation}
 Nevertheless, because $\chan_{j,k_j}$ is a CPTP channel, it can be simulated by sampling from a proper probability distribution as defined in step (c)i. of Algorithm \ref{fig:static}. We stress that while $\Phi_{\chan_{j,k_j}}$ is subject to the constraint ensuring that each $\chan_{j,k_j}$ is trace-preserving, this need not apply to the individual terms $\op{\psi_{l,j,k_j}}$ in its decomposition. Indeed, any decomposition of a non-unital channel cannot be represented as a linear combination of unitary Clifford operations alone. While Clifford gates are represented in the distribution by maximally entangled states, product states correspond to non-trace-preserving maps involving projections. For example, $\Phi_{\curlyT}= \op{01}$ corresponds to the single Kraus operator $\curlyT = \op{0}{1}$ which can be seen as a $Z$-measurement post-selected on the ``$\ket{1}$'' outcome, followed by an $X$ gate. Under these well-defined stabiliser operations, the state update in step (c)iii. can be carried out efficiently as per the Gottesman-Knill theorem \cite{Gottesman1997,GKtheorem}. More generally, the state update corresponding to any pure Choi state $\ket{\psi_{l,j,k_j}}$ can be modelled as a post-selected Bell measurement \cite{Gottesman1999}:
 \begin{equation} \ket{\phi_{j}}^A \otimes \ket{\Omega_n}^{B|C} \propto (\idn{n}^A\otimes \op{\Omega_n}^{B|C})\ket{\psi_{l,j,k_j}}^{AB} \otimes \ket{\phi_{j-1}}^C.
 \label{eq:stateupdate}
 \end{equation}
Note that the new state $\ket{\phi_{j}}$ is guaranteed to be pure, since the projection of the system BC onto a Bell state removes any correlation across the partition $A|BC$. 
We emphasise that the probability for carrying out this update depends not just on the known prefactors $|q_{l,j,k_j}|$, but on $\Tr[\curlyT_{l,j,k_j}(\op{\phi_{j-1}})]$, which are input state dependent. Computing this trace amounts to evaluating the overlap between $(\idn{n}\otimes \op{\Omega_n})$ and some stabiliser state, as per equation \eqref{eq:stateupdate}, which can be done efficiently using the stabiliser tabulex method \cite{Gottesman1997,GKtheorem}. However, since $\ket{\phi_{j-1}}$ will have been chosen randomly in the previous step of the algorithm, these traces cannot be calculated ahead of time. This means that our algorithm involves additional per-sample runtime overhead in order to calculate $\{p_l^{(j)}\}$. However, there are a finite and bounded number of such calculations since each circuit element acts on a small number of qubits. The key point is that the constraint on $\chan_{j,k_j}$ ensures that, despite being composed of non-trace-preserving elements, the distribution over the outcomes $\{p_l^{(j)}\}$ for any given state forms a proper probability distribution. Consequently, the intermediate sampling in step (c)i. does not increase the variance, so the sample complexity of the simulator depends only on $\prod_{j}\mathcal{R}_*(\mathcal{E}_j)^2$.

Suppose we want to compare the runtime of our simulator for a particular circuit with that of the Bennink et al. algorithm \cite{Bennink2017}. Their decompositions are in terms of $\CPR$, the set of Cliffords and Pauli reset channels, and an associated cost function is
	\begin{equation}
		\mathcal{R}_{\CPR} (\chan) = \min_{\Lambda_j \in \CPR}\qty{ ||p||_1: \sum p_j \Lambda_j = \chan},
	\end{equation}
	The sample complexity for simulating a given circuit element $\chan_j$ is proportional to $\mathcal{R}_{\CPR}(\chan_j)^2$. Since $\CPR \subseteq \SPOmth{n,n}$, it must be the case that $\mathcal{R}_* \leq \mathcal{R}_{\CPR}$, potentially leading to lower simulation sample complexity if there exist channels with $\mathcal{R}_* < \mathcal{R}_{\CPR}$. We give here a simple toy example demonstrating a significant advantage to our static simulator.

Consider the single-qubit CPTP map $\Lambda_H$ defined by a $Z$-measurement followed by a Hadamard gate conditioned on the ``-1'' outcome. This has Kraus representation:
\begin{equation}
	K_1 = \op{0}, \quad K_2 = \op{-}{1}.
\end{equation}  
This is clearly a completely stabiliser-preserving map, so has channel robustness $\mathcal{R}_*\qty(\Lambda_H) = 1$. For a single qubit, $\CPR$ consists of the 24 Clifford gates, and 6 Pauli reset channels. Using this set, we calculate $\mathcal{R}_{\CPR}\qty(\Lambda_H) = 2$. Since $\Lambda_H \in \SPOmth{1,1}$, this confirms that $\CPR$ is a strict subset of the completely stabiliser-preserving channels, and indicates that $\mathcal{R}_{\CPR}$ is not a monotone under stabiliser operations. We also note that the calculated value is larger than the robustness of magic for any single-qubit state, despite $\Lambda_H$ being a stabiliser operation. For a circuit containing M uses of the channel $\Lambda_H$, the samples required for a CPR simulator would be proportional to $\mathcal{R}_{\CPR}(\Lambda_H)^{2M} = 4^M $. But for our simulator, $\Lambda_H$ can be simulated efficiently, as $\mathcal{R}_*\qty(\Lambda_H)^{2M} = 1$. 

While the above example is quite artificial, a reduction in sample complexity is also achieved for channels where $\mathcal{R}_* (\chan) > 1$, but is strictly smaller than $\mathcal{R}_{\CPR}(\chan)$. Given a circuit decomposed as $\chan = \mathcal{E}_L \circ \ldots \circ \mathcal{E}_2 \circ \mathcal{E}_1$ , the sample complexity for the CPR simulator would be proportional to $\prod_j^L \mathcal{R}_{\CPR}(\chan_j)^2$. It is always the case that $\mathcal{R}_*(\chan_j) \leq \mathcal{R}_{\CPR}(\chan_j)$, so the sample complexity for our simulator will never be greater. Suppose we find that there are $M$ circuit elements such that $\mathcal{R}_*(\chan_j)/\mathcal{R}_{\CPR} (\chan_j) \leq k $ for some constant $0 < k < 1$. Then we would find that using our simulator gives a reduction in sample complexity by a factor of $k^{2M}$. While our simulator sometimes incurs a modest increase in the runtime per sample, this must be weighed against a reduction in sample complexity that is exponential in the number of circuit elements where $\mathcal{R}_*(\chan_j) < \mathcal{R}_{\CPR}(\chan_j) $. The obvious next question is whether there are any natural non-trivial examples where this happens. We show in Section \ref{sec:numerical} that gate sequences subject to amplitude-damping noise provide one such case.

We also note that calculation of optimal $\CPR$ decompositions is only tractable for one- and two-qubit circuit elements, as the three-qubit case already involves a linear program with nearly 93 million variables \cite{Bennink2017}. For the most general quantum channels, we encounter a similar problem, as for three-qubit circuit elements, we in principle need to optimise over six-qubit stabiliser states. However, in Appendix \ref{app:diagonal} we show that for diagonal channels the problem can be greatly simplified, and the problem becomes tractable for operations on up to five qubits. This allows our algorithm to take advantage of the submultiplicativity of channel robustness; for example for diagonal channels where $\mathcal{R}(\chan^{\otimes n}) < \mathcal{R}(\chan)^n$, it is advantageous to compose $n$ single-qubit circuit elements together as a single $n$-qubit circuit element, before running the linear program. We will see in Section \ref{sec:numerical} that this strategy is useful for single-qubit $Z$-rotations.

\subsection{Dynamic Monte Carlo\label{subsec:algorithm_convex}}

In the previous simulator, all convex optimisations are calculated in advance.  However, we have found examples of channels where $\mathcal{C}(\mathcal{E})<\mathcal{R}_* (\mathcal{E})$.  For such channels, and for any stabiliser state $\rho$, the robustness of the output state $\mathcal{E}(\rho)$ will always be less than the $\ell_1$-norm of the decomposition of $\mathcal{E}$ into stabiliser-preserving  CPTP channels.  Our next simulator takes advantage of this, and we present pseudocode in Algorithm ~\ref{SimulatorB}. 
\begin{algo}[ht]
			\raggedright
			{\bf Input:} A sequence of $L$ quantum channels $\mathcal{E}_j$ and number of samples $N$. \\
			{\bf Output:} An estimate of an expectation value.
\begin{enumerate}
	\item Set $i \leftarrow 0$ and $T \leftarrow 0$;
		\item For $i \leq N$ do: 
		\begin{enumerate}	
	\item  $i \leftarrow i+1$; 
	\item Set $\ket{\phi_0}=\ket{0}^{\otimes n}$; 
	\item  $R\leftarrow 1$;
	\item For $1 \leq j \leq L $ do
	\begin{enumerate}
		\item The channel $\mathcal{E}_j$ acts non-trivially on only $m$ qubits.  Partition the qubits into three sets $A|B|C$ where $A$ is the set acted on by $\mathcal{E}_j$, $B$ is a set of any other $m$ qubits; and $C$ comprises the remaining qubits;
		\item Find a Clifford $U= \id^A \otimes U^{BC} $ that is local w.r.t $A| (B \cup C)$ such that $U \ket{\phi_{j-1}}=\ket{\phi_{j-1}^{AB}} \otimes \ket{\phi_{j-1}^{C}} $.  This uses the efficient algorithm of Ref.~\cite{Fattal2004}.
		\item Find $2m$-qubit density matrix $\rho^{AB}_j=\mathcal{E}_j\otimes \id \qty( \ket{\phi_{j-1}^{AB}}\bra{\phi_{j-1}^{AB}}   )$;
		\item Solve convex optimisation to find $\mathcal{R}(\rho_j^{AB})$ and use optimal decomposition to build a quasiprobability distribution ;
		\item Sample from the renormalised quasiprobability distribution to choose a stabiliser state $\ket{\phi'^{AB}_j}$, and then set ${\ket{\phi_{j}} =  U^\dagger \ket{\phi'^{AB}_j} \otimes \ket{\phi_{j-1}^C}}$;
		\item Replace $R \leftarrow R \times \mathcal{R}(\rho_j^{AB}) \times \lambda$ where  $\lambda = \pm 1$ and denotes the phase of the sampled quasiprobability.		
		\item increment $j \leftarrow j+1$ and loop;
	\end{enumerate}	
	\item Evaluate $E=\bra{\phi_L}  Z \ket{\phi_L}$
	\item  $T \leftarrow T + (R \times E)$.
     \end{enumerate}	
\end{enumerate}
			{\bf Return:} $T/N$.
	\caption{\raggedright A classical simulator with sample complexity per step upper-bounded by $\mathcal{C}^2$.}
	\label{SimulatorB}
\end{algo}

As in the previous simulator, we can represent the trajectory through the circuit by a vector $\vec{k} $, such that the output of the true quantum circuit would be $\rho = \sum_{\vec{k}} q_{\vec{k}} \sigma_{\vec{k}} $. The major difference is that $q_{\vec{k}}$ cannot be written in the form of equation \eqref{Product_distro} with $p_j$ independent of $\vec{k}$, as the quasiprobabilities for each intermediate decomposition will depend on the stabiliser state sampled in the previous step. Consider what happens for the $j$th circuit element. After step (d)ii. of the algorithm, we have some stabiliser state $\sigma_{\vec{k}_j}$, where $\vec{k}_j$ labels the trajectory up to the $j$th element. Here we do not assume $\vec{k}_j$ is a binary vector; instead the elements of the vector label each pure stabiliser state. After steps (d)iii. and iv. we have a non-stabiliser state $\rho_{\vec{k}_j} = (\chan_j \otimes \id)(\sigma_{\vec{k}_j}) $, decomposed as:
\begin{equation}
    \rho_{\vec{k}_j} = \sum_{\vec{k}_{j+1}} q_{\vec{k}_{j+1}} \sigma_{\vec{k}_{j+1}}.
\end{equation}
Here the summation is over all $(j+1)$-step trajectories consistent with the previous $j$-step trajectory labelled by $\vec{k}_j$. In steps (d)v. and vi. we will choose the stabiliser state $\sigma_{\vec{k}_{j+1}}$ with probability $|q_{\vec{k}_{j+1}}|/ \mathcal{R}(\rho_{\vec{k}_j})$ and the variable $R$ picks up a factor $ \lambda_{\vec{k}_{j+1}} \mathcal{R}(\rho_{\vec{k}_j})$, where $\lambda_{\vec{k}_{j+1}}$ is the sign of the corresponding quasiprobability.

The final state after the full sequence of quantum operations may be written:
\begin{equation}
    \rho = \sum_{\vec{k}} p_{\vec{k}} R_{\vec{k}} \sigma_{\vec{k}}, \quad \text{where} 
    \quad p_{\vec{k}} = \prod_{j=0}^{L-1} \frac{|q_{\vec{k}_{j+1}}|}{\mathcal{R}(\rho_{\vec{k}_j})},
    \quad R_{\vec{k}} = \prod_{j=0}^{L-1} \sign{q_{\vec{k}_{j+1}}} \mathcal{R}(\rho_{\vec{k}_j}).
\end{equation}
The true expectation value for the observable $Z$ is given by:
\begin{equation}
    \Tr[Z \rho ] = \sum_{\vec{k}} p_{\vec{k}}  R_{\vec{k}} E_{\vec{k}}, \quad \text{where} \quad E_{\vec{k}} = \Tr[Z \sigma_{\vec{k}}].
\end{equation}
A key difference with the previous simulator is that we never explicitly calculate the full distribution $p_{\vec{k}}$. Nevertheless, each time the simulator samples, it produces output $R_{\vec{k}} E_{\vec{k}}$  with probability $p_{\vec{k}}$. These probabilities exactly match the weightings in the above equation, so the simulator is an unbiased estimator. The number of samples required can be again derived using the Hoeffding inequalities, which depend on the maximum possible values of the output of each sample.  Each output is bounded by $|R_{\vec{k}} E_{\vec{k}}| \leq |R_{\vec{k}}|\leq \prod_j \mathcal{R}(\rho_{\vec{k}_j}) \leq \prod_j \mathcal{C}(\mathcal{E}_j)$.  Therefore, the sample complexity is upper-bounded by order $ \prod_j  \mathcal{C}(\mathcal{E}_j)^2$.

Notice that for every sample, $L$ convex optimisations are performed, as well as $L$ steps involving the algorithm of Fattal et al. \cite{Fattal2004}, which has runtime polynomial in the total number of qubits. If $\mathcal{C}(\mathcal{E})=\mathcal{R}_* (\mathcal{E})$ then we would simply not use this method so that the only convex optimisations are in the preprocessing.  However, $\mathcal{C}(\mathcal{E})$ now determines the sample complexity, so if $\mathcal{C}(\mathcal{E}) \ll \mathcal{R}_* (\mathcal{E})$ then the dynamic simulator may run much faster than the static simulator; here we have a trade-off of increase in per-sample runtime, versus a possibly exponential reduction in sample complexity.  Indeed, since the sample complexity is typically the bottleneck, this approach would lead to significant improvements for some quantum channels. In the following section we investigate which types of channels may lead to an advantage.

\section{\label{sec:numerical} Numerical results}
The numerical results in this section have been produced using code and data files available from the public repository detailed in Ref. \cite{channel_repo}.

\subsection{\label{subsec:amp_damp} Single-qubit rotation with amplitude damping}

Consider the setting discussed in Section \ref{sec:algorithms}, where a many-qubit circuit evolution $\chan$ is decomposed as a series of few-qubit circuit elements $\chan =\mathcal{E}_L  \ldots \circ \mathcal{E}_2 \circ \mathcal{E}_1 $. Many implementations of quantum algorithms can be expected to involve single-qubit rotations about some Pauli axis. In near-term devices, the circuit will be subject to noise. Consider a simple model of a noisy computation where a noise channel $\Lambda$ acts between each unitary gate $\mathcal{U}_j$, so the overall channel representing the circuit would be:
\begin{equation}
    \chan = \Lambda \circ \mathcal{U}_L \circ \ldots \circ \Lambda \circ \mathcal{U}_2 \circ \Lambda \circ \mathcal{U}_1.
    \label{eq:noisy_circuit}
\end{equation}
Here we study the simulation cost for a single step in such a computation, comprised of a single-qubit rotation and a noise channel. Note that for intermediate steps in a circuit decomposition such as \eqref{eq:noisy_circuit}, we have a choice of ordering. We can take the circuit elements to be either $\Lambda \circ \mathcal{U}_j$ or $\mathcal{U}_j \circ \Lambda$. These choices are equivalent in terms of the output of the simulation, but could lead to different sample complexity depending on the cost function used. We studied circuit elements made up of a single-qubit Pauli $X$-rotation $U(\theta) = \exp(i X \theta)$ composed with an amplitude damping channel $\Lambda_p$ with noise parameter $p$, defined by Kraus operators:
\begin{equation}
    K_1 = \begin{pmatrix}
        1 & 0 \\
        0 & \sqrt{1-p}
    \end{pmatrix} , \quad 
    K_2 = \begin{pmatrix}
        0 & \sqrt{p} \\
        0 & 0 
    \end{pmatrix}
\end{equation}
We calculated channel robustness $\mathcal{R}_*$, magic capacity $\mathcal{C}$, and the Bennink et al. \cite{Bennink2017} cost function $\mathcal{R}_{\CPR}$, for both $\Lambda_p \circ \mathcal{U}(\theta)$ and $\mathcal{U}(\theta) \circ \Lambda_p$, and for a range of values of $p$ and $\theta$. For noise $p=0.1$ we see that when the noise channel follows the gate, there is no difference between the three quantities (Figure \ref{fig:amp_damp}(i)). However, if the noise channel acts before the unitary, both our monotones show a reduced value, whereas $\mathcal{R}_{\CPR}$ increases. This suggests that for this noise model, the better strategy with respect to sample complexity would be to choose the ordering $\mathcal{U}(\theta) \circ \Lambda_p$, and  use one of our simulators (subject to the caveats mentioned in the previous section). In Figure \ref{fig:amp_damp}(ii) we show how different levels of noise affect the channel robustness. We do not plot capacity since we find that $\mathcal{R}_*(\chan) = \mathcal{C}(\chan)$ for this class of operation, up to solver precision. We also compare channel robustness with the Choi state robustness (Figure \ref{fig:amp_damp}(iii)). We find that $\RChoi{\chan} = \mathcal{R}_*(\chan)$ for $\theta$ up to approximately $\pi/16$, but  $\RChoi{\chan} < \mathcal{R}_*(\chan)$ for larger angles.

\begin{figure}[thbp]
    \centering
    \includegraphics[width=0.8\textwidth]{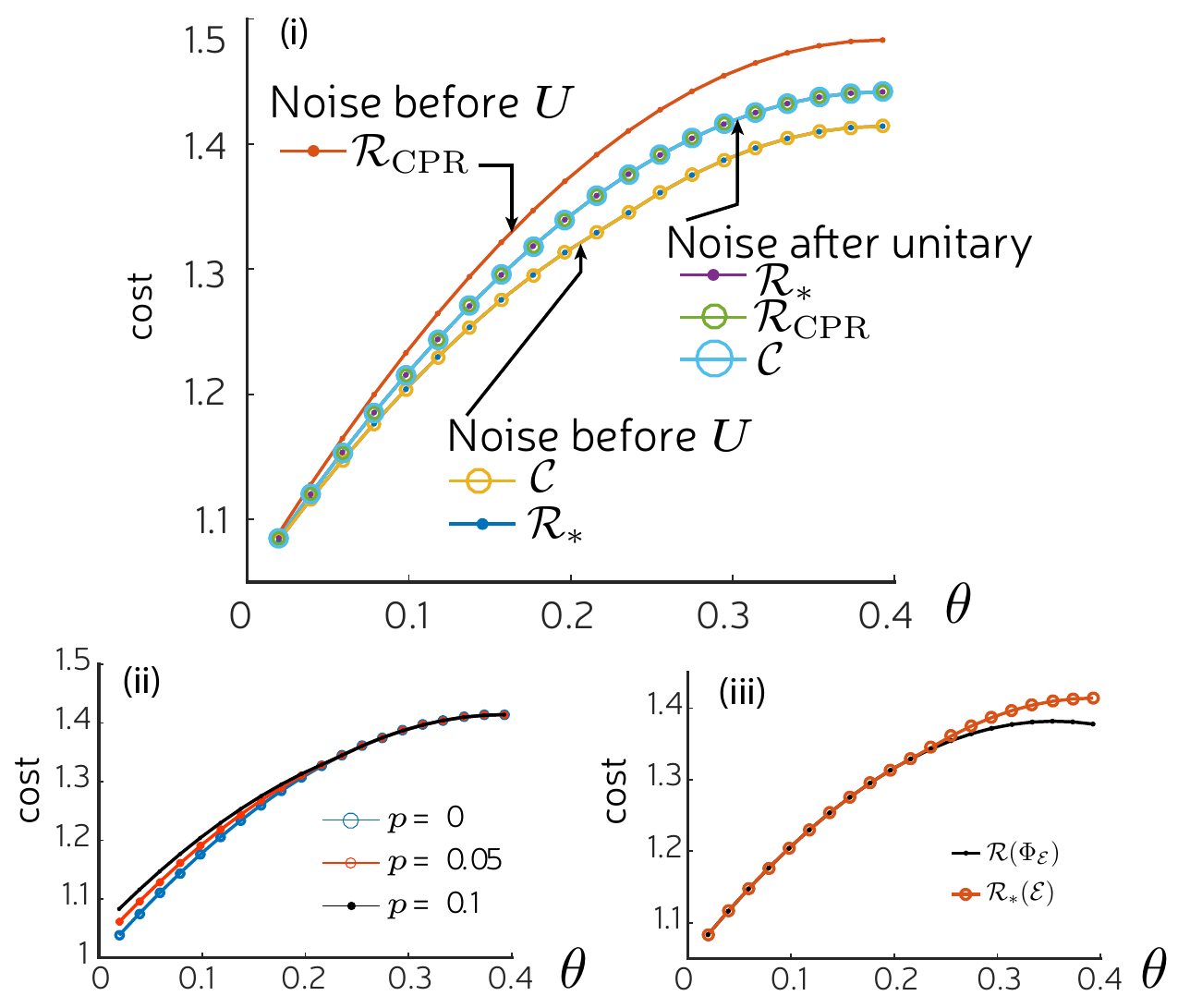}
    \caption{(i) Comparison of $\mathcal{C}(\chan)$, $\mathcal{R}_*(\chan)$ and $\mathcal{R}_{CPR}(\chan)$, where $\chan$ is a single-qubit $X$-rotation $\mathcal{U}(\theta) $ composed with an amplitude damping channel $\Lambda_p$. We consider both possible orderings: noise after unitary ($\Lambda_p \circ \mathcal{U}(\theta)$), and noise before unitary ($\mathcal{U}(\theta) \circ \Lambda_p$). (ii) $\mathcal{R}_*\qty(\mathcal{U}(\theta) \circ \Lambda_p)$ for several values of $p$. (iii) Comparison of channel robustness with robustness of Choi state for $\mathcal{U}(\theta) \circ \Lambda_p$ with $p=0.1$.}
    \label{fig:amp_damp}
\end{figure}

\subsection{Multiqubit phase gates}

Recall that from Theorem \ref{thm:sandwich} we know  $\RChoi{\chan} \leq \Mcap{\chan} \leq \Rstar{\chan}$. For the particular class of channel studied above, we saw numerically that $\Mcap{\chan} = \Rstar{\chan}$ up to solver precision in all cases investigated, and in the absence of noise the numerical results suggested $\RChoi{\chan} = \Mcap{\chan} = \Rstar{\chan}$. Moreover, we know from Theorem \ref{thm:sandwich} that all measures are equal for gates from the third level of the Clifford hierarchy. Under what conditions does this equality persist for multi-qubit operations? As explained earlier, to calculate each of our quantities for $n$-qubit channels, in general we must solve an optimisation problem over all $2n$-qubit stabiliser states. Since this problem is only tractable for up to $5$-qubit states, in practice we are limited to studying two-qubit channels, in the most general case. However, it turns out the problem can be greatly simplified for certain types of operation. In particular, Appendix \ref{app:diagonal} shows how the problem size can be reduced for channels diagonal in the computational basis, using a representation of stabiliser states in terms of affine spaces over binary vectors due to Dehaene and De Moor \cite{Dehaene2003,Gross2008}. MATLAB code to calculate our measures for this reduced problem is provided in the public repository given in Ref. \cite{channel_repo}. This allows us to calculate values for diagonal operations on up to 5 qubits, which we present here.

As a special case we consider multicontrol phase gates of the form:
\begin{equation}
M_{t,n} = \mathrm{diag} ( \exp(  i \pi / 2^t ) , 1 , 1 , \ldots , 1 ), \quad t \in \mathbb{Z}
 \label{eq:CCphase}
\end{equation}
where  $n$ denotes the number of qubits. We note that by convention, controlled-phase gates typically apply the phase to the all-one state $\ket{1^n}$, where $1^n = (1,\ldots,1)^T$, but the form given above is Clifford-equivalent to the conventional version, and is more convenient for the techniques used in Appendix \ref{app:diagonal}. The family includes familiar gates such as $CZ$ ($t=0$, $n=2)$, $CCZ$ ($t=0$, $n=3$), multicontrol-$S$ ($t=1)$ and multicontrol-$T$ ($t=2$). 

The main findings were that the inequalities are tight for the $n=2$ and $n=3$ cases, but that this does not persist for larger system sizes (Figure \ref{fig:multicontrol_comp}). The $t=0$ case (the family of multicontrol-$Z$ gates) turns out to be a special case (Figure \ref{fig:multicontrol_comp}, left panel). Here we find equality for all three quantities up to $n=4$. For the $t=0$, $n=5$ case, $\RChoi{M_{0,5}} = \Mcap{M_{0,5}}$ holds, but $\Rstar{M_{0,5}}$ is strictly greater than both. Note also that for $t=0$, all three quantities increase with each increment in $n$.

\begin{figure}[thbp]
    \centering
    \includegraphics[width=1\textwidth]{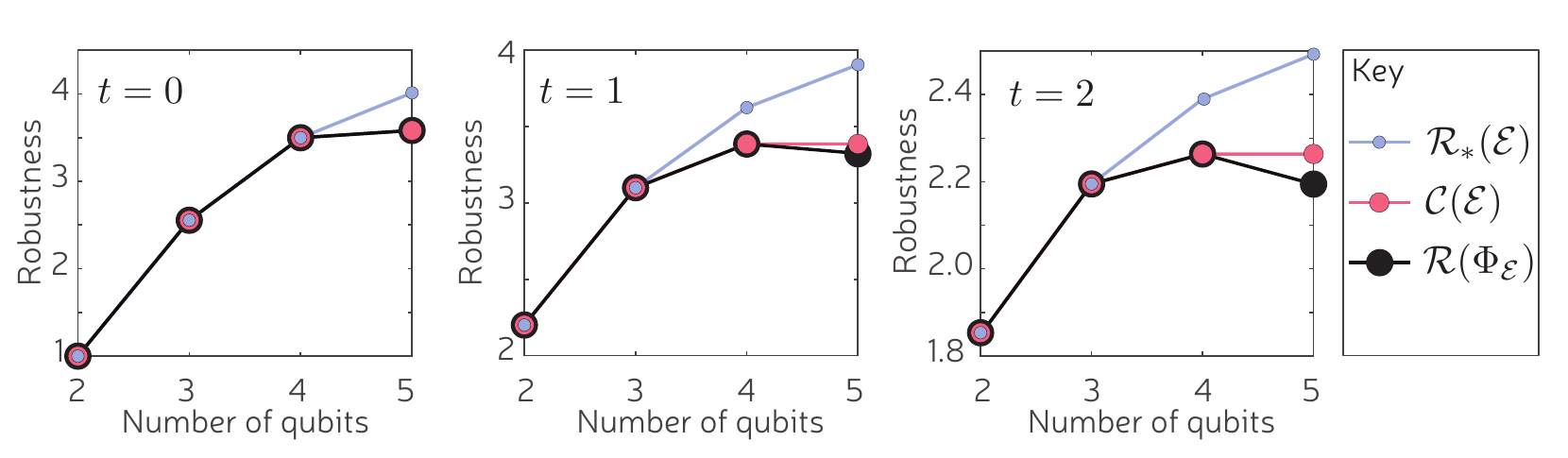}
    \caption{Comparison of quantities for multicontrol phase gates (see equation \ref{eq:CCphase}). Left: Multicontrol-$Z$ gates (t=0). Middle: Multicontrol-$S$ gates (t=1). Right: Multicontrol-$T$ gates (t=2).}
    \label{fig:multicontrol_comp}
\end{figure}

\begin{figure}[thbp]
    \centering
    \includegraphics[width=\textwidth]{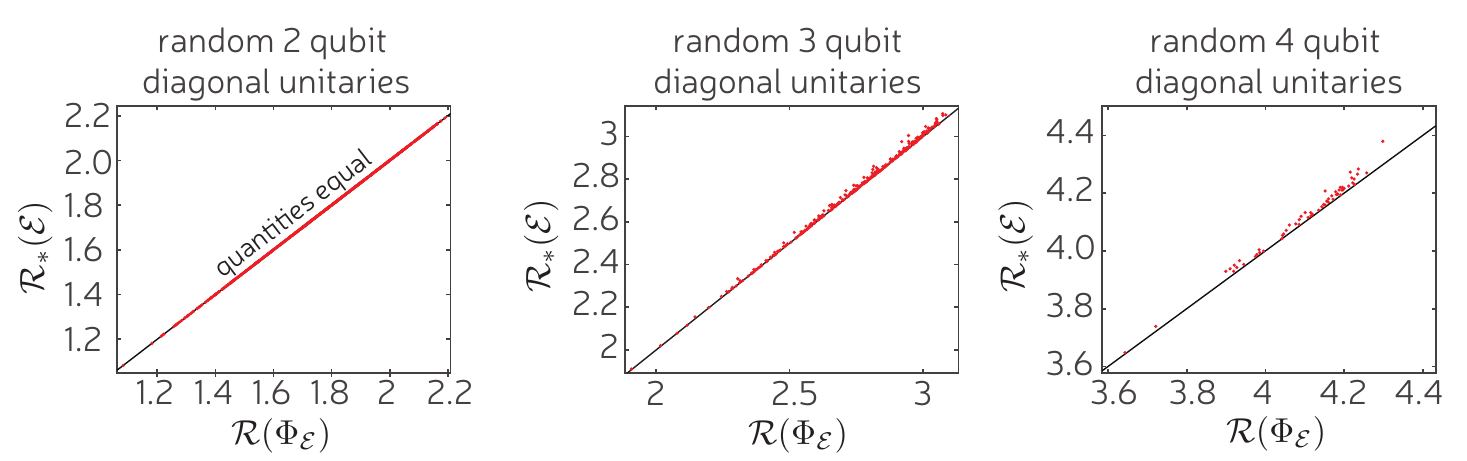}
    \caption{Channel robustness against robustness of Choi state for random $n$-qubit diagonal gates, up to $n=4$. Black line indicates $\mathcal{R}_* = \RChoi{}$. Each red dot represents the data point for an individual gate. Fewer points were calculated for larger $n$ due to the increased time to calculate each value. 1000 data points were calculated for $n=2$, $300$ for $n=3$, and $60$ for $n=4$.}
    \label{fig:random_diag_plots}
\end{figure}
The families of gates with $t>0$ follow a pattern qualitatively similar to each other. The results for the $t=1$ (multicontrol-$S$) and $t=2$ (multicontrol-$T$) cases are shown in the middle and right panels of Figure \ref{fig:multicontrol_comp}. For $n=4$, $t>0$, the same situation holds as for $n=5$, $t=0$, as we find $\RChoi{M_{t,4}} = \Mcap{M_{t,4}} < \Rstar{M_{t,4}}$. At $n=5$, all three quantities separate. In contrast with the multicontrol-$Z$, we see that $\RChoi{M_{t,n}}$ decreases as we go from four to five qubits, while  $\Mcap{M_{t,n}}$ levels off. We see similar behaviour for all non-zero values of $t$ investigated numerically. Our current techniques limit us to five-qubit gates, but we have reason to believe that the capacity will remain level for $n>5$, and we make the following conjecture, which we justify more fully in Appendix \ref{app:dim_affine}.
\begin{conjecture}
    For any fixed $t$, the maximum increase in robustness of magic for $M_{t,n}$ is achieved at some finite number of qubits $n=K$ by acting on the state $\ket{+}^{\otimes K}$. Therefore $\Mcap{M_{t,n}} = \robmag{M_{t,K}\ket{+}^{\otimes K}}{}$ for all $n\geq K$.
    \label{conj:fourqubitMulticontrol}
\end{conjecture}

We also numerically investigated the robustness of diagonal unitaries 
\begin{equation}
	 U = \sum_{x} e^{i \theta_x} \vert x \rangle \langle x \vert ,
	\end{equation}
with $\theta_x$ chosen uniformly at random.  We were particularly interested in understanding when $\mathcal{R}(\Phi_\mathcal{E}) \leq C(\mathcal{E}) \leq \mathcal{R}_*(\mathcal{E})$ is tight or loose.  In Figure ~\ref{fig:random_diag_plots} we compare the Choi robustness with the channel robustness.  For every 2-qubit gate tested we observed that $\mathcal{R}(\Phi_\mathcal{E})= \mathcal{R}_*(\mathcal{E}) $ up to numerical precision.  Whereas, for 3 and 4 qubit gates we typically saw that $\mathcal{R}(\Phi_\mathcal{E})< \mathcal{R}_*(\mathcal{E})$, though the gap is not often large.  While the difference is slight for a single gate, these quantities influence the rate of exponential scaling when considering $N$ uses of such a unitary and will lead to a large gap for modest $N$.

We also compared the Choi robustness with the magic capacity but do not plot this data as it was equal for every random instance we observed.  This is curious since in Figure ~\ref{fig:multicontrol_comp} we clearly see that there do exist diagonal gates, the multicontrol phase gates, for which there is a gap between the Choi robustness and the magic capacity.  While such gates exist, our random sampling of diagonal gates does not tend to provide such examples. We discuss this further in Appendix \ref{app:dim_affine}.

Finally, we are also interested in the \emph{normalised} channel robustness for single-qubit gates $U$, defined as $\qty[\mathcal{R}_*(U^{\otimes n})]^{1/n}$. This allows us to quantify the per-gate savings in sample complexity that can be achieved by grouping single-qubit rotations in $n$-qubit blocks. In Figure \ref{fig:regularised} we present results for qubit $Z$-rotations $U = \exp[iZ \theta]$, up to four qubits. We find that strict submultiplicativity is observed for all values of $\theta$, with significant reductions between the $n=2$ and $n=4$ cases for a wide range of angles.

\begin{figure}[thbp]
    \centering
    \includegraphics[trim={0 9.3cm 0 9.2cm},clip,width=0.6\textwidth]{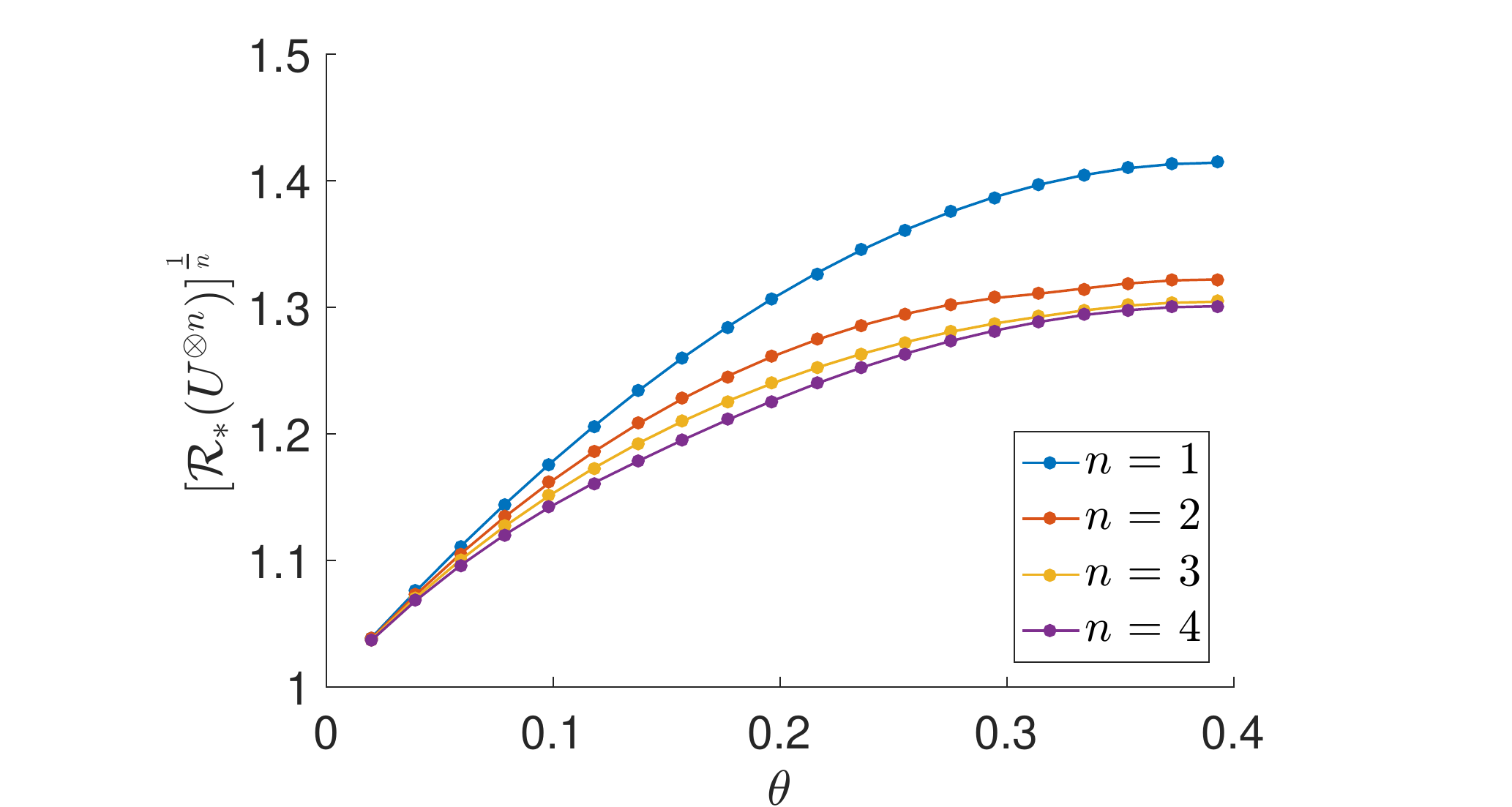}
    \caption{Normalised channel robustness $[\mathcal{R}_*(U^{\otimes n})]^{1/n}$ plotted for $Z$-rotations $U(\theta) = \exp[i Z \theta]$ and for $n$ qubits, up to $n=4$.}
    \label{fig:regularised}
\end{figure}

\section{Conclusion}
We have presented two new magic monotones for general quantum channels: the magic capacity $\mathcal{C}$, which quantifies the ability of a channel to generate magic, and the channel robustness $\mathcal{R}_*$, which is related to finding the minimal quasiprobability decomposition of a channel into stabiliser-preserving CPTP maps. Each of these monotones is directly related to the sample complexity for an associated Monte Carlo-type classical simulation algorithm. We found that for certain quantum channels, our static simulator would lead to a exponentially better sample complexity as compared to that for the algorithm due to Bennink et al. \cite{Bennink2017}. In particular we found a reduction in sample complexity for the case of a sequence of single-qubit rotations subject to amplitude-damping noise. Since our decompositions can be calculated for up to five qubits in the case of diagonal operations, our static simulator is also able to take advantage of the submultiplicativity of channel robustness under tensor product and composition.

For some channels, further improvements in sample complexity are possible using a different simulator that is related to the capacity.  That is, we found the capacity can be strictly less than channel robustness for certain multi-qubit entangling gates.  This simulator has to introduce on-the-fly convex optimisation, however, so each sample will be more difficult to obtain.

For simulation of realistic quantum devices involving many qubits, one would need to decompose the circuit into a sequence of operations on smaller number of qubits, as described in Section \ref{sec:algorithms}, in order that the associated optimisation problems are tractable.  It is a non-trivial problem to decide what is the optimal way to block together the few-qubit operations making up a given many-qubit circuit: we saw in Section \ref{sec:numerical} that the ordering of operations can make a difference to the sample complexity. We leave this problem for a future work.

Since our monotones are submultiplicative under tensor product and compositions, it is generally preferable to combine circuit elements where possible. In practice, to calculate our monotones, circuit elements can involve up to two qubits for the most general case, or at most five qubits for diagonal operations, where we can make use of the techniques described in Appendix \ref{app:diagonal}. In Ref. \cite{Heinrich2018}, Heinrich and Gross show that robustness of magic can be calculated for up to 10 copies of the resource states employed in the standard magic state model of fault-tolerant quantum computation. Their methods rely on both the permutation symmetry due to having multiple copies, and the stabiliser symmetries of the magic states considered. Another direction for future work could therefore be to investigate whether (and under what conditions) similar techniques can be applied to the channel picture to increase the number of qubits that can be involved in each circuit element.

\begin{acknowledgments}
This work was supported by the Engineering and Physical Sciences Research Council [grant numbers EP/P510270/1 (J.R.S.) and EP/M024261/1 (E.T.C.)] .

We would like to thank Mark Howard for many useful discussions and assistance with linear optimisation. We would also like to thank Ryuji Takagi for interesting discussions on stabiliser-preserving operations. Thank you to Hakop Pashayan and Dan Browne for helpful discussions on classical simulation and the discrete Wigner function. Finally we would like to thank the anonymous reviewers for their comments and suggestions on the manuscript.
\end{acknowledgments}


\bibliography{ms}

\appendix
\renewcommand{\theequation}{A\arabic{equation}}
             \renewcommand{\thetheorem}{A\arabic{theorem}}
             \renewcommand{\thelemma}{A\arabic{lemma}}
             \renewcommand{\theobs}{A\arabic{obs}}
             \renewcommand{\theconjecture}{A\arabic{conjecture}}
\section{Properties of robustness of the Choi state \label{app:choiR}}
Here we confirm that the robustness of the Choi state, $\RChoi{}$ has the properties convexity and submultiplicativity under tensor product. We then give an example to show that it is not submultiplicative under composition.

\noindent\textbf{Convexity:} 
This follows immediately from convexity of robustness of magic. Consider a real linear combination of $n$-qubit channels: $\mathcal{E} = \sum_k q_k \mathcal{E}_k$. The Choi state for $\chan$ is:
\begin{align}
    \choistate{\chan} & = \outstate{\chan}{\op{\Omega_n}}{n} \\
                        & = \sum_k q_k \outstate{\chan_k}{\op{\Omega_n}}{n}  = \sum_k q_k \choistate{\chan_k},
\end{align}
where in the last line we identified $\outstate{\chan_k}{\op{\Omega}}{n}$ as the Choi state for $\chan_k$. Then by convexity of robustness of magic:
\begin{equation}
     \robmag{\choistate{\chan}}{} \leq \sum_k \abs{q_k}\robmag{\choistate{\chan_k}}{},
\end{equation}
which shows $\mathcal{R}\qty(\Phi_\chan)$ is convex in $\chan$.

\noindent\textbf{Submultiplicativity under tensor product:} The maximally entangled state $\ket{\Omega_{n+m}}^{AA'|BB'}$ as defined by equation \eqref{eq:CJdefn} in the main text can be factored as $\ket{\Omega_{n+m}}^{AA'BB'} = \ket{\Omega_n}^{A|B} \ket{\Omega_m}^{A'|B'}$. So the Choi state for a channel $\chan^{AA'} = \chan^A \otimes \chan'^{A'}$, where $\chan^A$ and $\chan'^{A'}$ are respectively $n$-qubit and $m$-qubit channels, can be written:
\begin{align}
    \choistate{\chan} & = \outstate{\chan^A \otimes \chan'^{A'}}{\op{\Omega_{n+m}}^{AA'|BB'}}{n+m} \\
                        & = \outstate{\chan^A}{\op{\Omega_n}^{A|B}}{n} \otimes  \outstate{\chan'^{A'}}{\op{\Omega_m}^{A'|B'}}{m}  = \choistate{\chan^A} \otimes \choistate{\chan'^{A'}}.
\end{align}
Then by submultiplicativity of robustness of magic for states, we have:
$$
\robmag{\choistate{\chan^A \otimes \chan'^{A'}}}{} \leq \robmag{\choistate{\chan^A}}{} \robmag{\choistate{\chan'^{A'}}}{},
$$
which is the desired property.

\noindent\textbf{Failure of submultiplicativity under composition:} Let $\chan_1$ be the single-qubit $Z$-reset channel defined by Kraus operators $\qty{\op{0}, \op{0}{1}}$, and let $\chan_2$ be the conditional channel defined by $\qty{\op{T}{0}, \op{1}}$, where $\ket{T}=T\ket{+}$. These channels respectively have Choi states $\choistate{\chan_1}=\op{0} \otimes \frac{\id}{2}$,  and $\Phi_{\chan_2} = \frac{1}{2}\qty(\op{T0} + \op{11} )$, with robustness of magic $\RChoi{\chan_1} = 1$ and $\RChoi{\chan_2} \approx 1.207$.

The composed channel $\chan_2 \circ \chan_1$ has a Kraus representation $\qty{\op{T}{0},\op{T}{1}}$, and so has a Choi state $\Phi_{\chan_2 \circ \chan_1} = \op{T} \otimes \frac{\id}{2}$, with $\RChoi{\chan_2 \circ \chan_1} \approx 1.414 > \RChoi{\chan_2}\RChoi{\chan_1}$. So it is not the case that the robustness of the Choi state is submultiplicative under composition.

More intuitively, such counterexamples arise for channels $\chan$ where the stabiliser state $\ket{\phi_*}$ that results in maximal final robustness $\mathcal{R}\qty[\outstate{\chan}{\op{\phi_*}}{n}]$ is not the maximally entangled state $\ket{\Omega_n}$, as then we can always boost the output robustness by using a stabiliser-preserving operation to prepare $\ket{\phi_*}$ before applying $\chan$.

\section{Properties of channel robustness \label{app:channel_robustness}}

\noindent\textbf{Faithfulness:} Suppose $\chan$ is an $n$-qubit CPTP map. There are two cases:

(i) $\chan \in \SPOmth{n,n}$. In this case, $\Phi_\chan$ is itself a mixed stabiliser state, and since $\chan$ is trace-preserving it satisfies $\Tr(\Phi_\chan) = \idn{n}/2^n$. So $\Phi_\chan$ is already trivially a decomposition of the correct form, with $p=0$, so that $R_* (\chan) = 1 + 2p = 1$.

(ii) $\chan \notin \SPOmth{n,n}$. Then by faithfulness of robustness of magic (Lemma \ref{thm:choi_faithfulness} in the main text), $\Phi_\chan$ has $\RChoi{\chan} > 1$. Since the definition of $\mathcal{R}_*$ is a restriction of $\RChoi{}$, it must be the case that $\RChoi{\chan} \leq \Rstar{\chan}$. Therefore $\mathcal{R}_*(\chan)>1$.

\noindent\textbf{Convexity:} 

Suppose we have a set of Choi states $\Phi_{\chan_j}$ corresponding to channels $\chan_j$, with optimal decompositions:
\begin{equation}
\Phi_{\chan_j} = (1 + p_j) \rho_{j_+} - p_j \rho_{j_-},
\end{equation}
where each $\rho_{j_\pm}$ separately satisfies the condition $\Tr_A(\rho_\pm) = \frac{\idn{n}}{2^n}$, so that $\mathcal{R}_* (\chan_j) = 1 + 2 p_j$. Now take a real linear combination of such channels:
\begin{equation}
\chan = \sum_i q_i \chan_i = \sum_{j\in P} q_{j} \chan_{j} +  \sum_{k \in N} q_{k} \chan_{k} ,
\end{equation}
where $P$ is the set of indices such that $q_{j}\geq 0$, and $N$ is the set such that $q_{k}< 0$. We assume that $\sum_i q_i = 1$ so that the trace of $\Tr(\Phi_\chan) = 1$. Then the corresponding Choi state for channel $\chan$ is:
\begin{align}
\Phi_\chan & =  \sum_{j \in P} q_{j} \Phi_{\chan_j} -  \sum_{k\in N} \abs{q_k} \Phi_{\chan_k} \\
			& = \sum_{j\in P} q_{j} \qty[(1 + p_j) \rho_{j_+} -  p_{j} \rho_{j_-}]  - \sum_{k\in N} |q_{k}| \qty[(1 + p_{k}) \rho_{k_+} -  p_{k} \rho_{{k}_-}] \\
           & = \qty( \sum_{j \in P} q_{j} (1 + p_{j}) \rho_{j_+} + \sum_{k \in N} |q_{k}| p_k \rho_{k_-}) 
            - \qty(\sum_{j \in P} q_{j} p_{j} \rho_{j_-} +  \sum_{k \in N} |q_{k}| (1 + p_{k}) \rho_{k_+} ).
\end{align}
Note that the terms inside the brackets all have positive coefficients, hence we can interpret as non-normalised mixtures over stabiliser states. To normalise them we can define:
\begin{equation}
\widetilde{\rho}_+ = \frac{\sum_{j \in P} q_{j} (1 + p_{j}) \rho_{j_+} + \sum_{k \in N} |q_{k}| p_k \rho_{k_-}}{\sum_{j \in P} q_{j} (1 + p_{j}) + \sum_{k \in N } |q_{k}| p_k},
\end{equation}
and
\begin{equation}
\widetilde{\rho}_- = \frac{\sum_{j \in P} q_{j} p_{j} \rho_{j_-} +  \sum_{k \in N} |q_{k}| (1 + p_k) \rho_{k_+} }{\sum_{j \in P} q_{j} p_{j} +  \sum_{k \in N} |q_k| (1 + p_k)  }.
\end{equation}
Then writing:
\begin{equation}
\widetilde{p} = \sum_{j \in P} q_{j} p_{j} +  \sum_{k \in N} |q_k| (1 + p_k) ,
\end{equation}
one can check that:
\begin{equation}
1 + \widetilde{p} = \sum_{j \in P} q_{j} (1 + p_{j}) + \sum_{k \in N} |q_k| p_k.
\end{equation}
This allows us to rewrite the Choi state as:
$
\Phi_\chan = (1+ \widetilde{p}) \widetilde{\rho}_+ - \widetilde{p} \widetilde{\rho}_-.
\label{eq:Rstardecomp}
$
Since $\widetilde{\rho}_\pm$ are convex mixtures over stabiliser states satisfying $\Tr_A(\rho_{j\pm}) = \frac{\idn{n}}{2^n}$, they must satisfy the same condition. We also know that $\widetilde{p}\geq 0$, so it is clear that the decomposition is in the form required for the definition of $\mathcal{R}_*$, except that it is not necessarily optimised to minimise $1 + 2 \widetilde{p}$. So we have:
\begin{align}
\mathcal{R}_* \qty(\sum_j q_j \chan_j)  \leq 1 + 2 \widetilde{p} 
		& = \sum_{j \in P} q_{j} (1 + p_{j}) + \sum_{k \in N} |q_k| + \sum_{j \in P} q_{j} p_{j} +  \sum_{k \in N} |q_k| (1 + p_k) \\
        & = \sum_{j \in P} |q_{j}| (1 + 2 p_{j}) + \sum_{k \in N} |q_k| (1 + 2 p_k)\\
        & = \sum_i |q_i| \mathcal{R}_* (\chan_j),
\end{align}
which gives us the required result. 

\noindent\textbf{Invariance under tensor with identity:}
In Section \ref{sec:choi} of the main text we saw that $\Rstar{\chan^A \otimes \id} \leq \Rstar{\chan^A}$. We now complete the proof that $\Rstar{\chan^A \otimes \id} = \Rstar{\chan^A}$ by showing that $ \Rstar{\chan^A} \leq \Rstar{\chan^A \otimes \id} $.

Consider an optimal decomposition for $\Phi_{\chan^A \otimes \idn{m}} = (1+p') \rho'_{+} - p' \rho'_{-}$, such that $\mathcal{R}_* (\chan^A \otimes \idn{m}) = 1 + 2p'$, 
where  $\Tr_{AA'}(\rho_\pm) = \idn{n+m}/2^{n+m}$. Here we do not assume that $\rho'_\pm$ are products across the partition $AB|A'B'$, as was the case in equation \eqref{eq:productDecomp} in the main text. However, we have just seen that $\Phi_{\chan^A \otimes \idn{m}} $ can be written as a product, so that by tracing out systems $A'B'$ we obtain:
\begin{equation}
\Phi_{\chan^A} = (1 + p') \Tr_{A'B'}(\rho'_+) - p' \Tr_{A'B'}(\rho'_-).
\end{equation}
Partial trace of a stabiliser state remains a stabiliser state, so this is a stabiliser decomposition. We just need to check that the partial trace condition holds, so we want to show:
\begin{equation}
\Tr_A(\Tr_{A'B'} (\rho'_\pm)) = \Tr_{AA'B'} (\rho'_\pm) = \frac{\idn{n}}{2^n},
\end{equation}
but this is clearly the case from the fact that $\rho'_\pm$ were constrained such that $\Tr_{AA'}(\rho_\pm) = \idn{n+m}/2^{n+m}$. Hence again we have a valid, not necessarily optimal decomposition and:
\begin{equation}
{\mathcal{R}_*(\chan^A) \leq 1 + 2p' = \mathcal{R}_*(\chan^A \otimes \id^B)}.
\end{equation}
Combining with the inequality $\mathcal{R}_* (\chan^A \otimes \id^B) \leq \mathcal{R}_* (\chan^A)$, shown in the main text, we obtain the equality:
\begin{equation}
    \mathcal{R}_* (\chan \otimes \id) = \mathcal{R}_* (\chan) = \mathcal{R}_* (\id \otimes \chan).
\end{equation}

\section{Optimisation problem for channel robustness \label{app:channel_optimisation}}
In Howard and Campbell \cite{Howard2017}, the optimisation problem for calculating robustness of magic for states was cast as follows:
\begin{align*}
\text{\bf minimise}&\quad \lonevec{q}\\
\text{\bf subject to}& \quad A \vec{q} = \vec{b},
\end{align*}
where $\vec{q}$ is a vector of coefficients, $\vec{b}$ is the vector of Pauli expectation values for the target state $\Phi_\chan$, and $A$ is a matrix whose columns are the Pauli vectors for the stabiliser states. For $n$-qubit channels, we have $2n$-qubit Choi states, so the number of generalised Paulis is $N_P = 4^{2n}$, and the number of stabiliser states is $N_S = 2^{2n} \prod_{j=1}^{2n} (2^j + 1)$ \cite{Howard2017}. Then $\vec{b}$ has $N_P$ entries, $\vec{q}$ has $N_S$ entries, and the dimension of $A$ is $(N_P \times N_S)$. From this construction we can recover optimal decompositions of the form:
$
\Phi_\chan = \sum_j q_j \op{\phi_j}
$, where
$
\sum_j q_j = 1
$
and $\ket{\phi_j}$ are the pure stabiliser states.

We want to restrict the problem to decompositions of the form:
\begin{align}
\Phi_\chan = (1 + p) \rho_+ - p \rho_-,
\end{align}
where $p\geq 0$ and $\rho_\pm$ correspond to trace-preserving channels, and can in general be mixed. Rather than enumerating all the extreme points of the set of stabiliser states corresponding to maps in \SPO{n,n}, it is more convenient to retain the same $A$ matrix and modify the constraints. We still need to start from a finite set of extreme points, i.e. pure stabiliser states, so first rewrite as:
\begin{align}
\Phi_\chan &= \sum_j q_{j_+} \rho_j + \sum_j q_{j_-} \rho_j = \sum_j p_{j_+} \rho_j - \sum_j p_{j_-} \rho_j,
\end{align}
where $q_{j_+}$ are the positive quasiprobabilities, $q_{j_-}$ are the negative quasiprobabilities, and $p_{j_\pm} = |q_{j_\pm}|$. In the Pauli vector picture we can write this as  
$
\vec{b} = A \vec{p}_+ - A \vec{p}_-
$, 
where all the entries of $\vec{p}_\pm$ are non-negative. We define a new variable vector $\vec{p}$ which will have twice the length of the previous $\vec{q}$, i.e. $2N_S$ entries:
\begin{equation}
\vec{p} = \begin{pmatrix} \vec{p}_+ \\ \vec{p}_-,
\end{pmatrix}
\end{equation}
and define a new $(N_P \times 2N_S)$ matrix $A'$ in block form, 
$
A' = \begin{pmatrix}
 A & -A
\end{pmatrix}
$. 
Then we have:
\begin{equation}
A' \vec{p} = \begin{pmatrix}
 A & -A
\end{pmatrix} 
\begin{pmatrix} \vec{p}_+ \\ \vec{p}_-
\end{pmatrix} = A \vec{p}_+ - A \vec{p}_- = \vec{b} .
\end{equation}
So now we need to minimise $\lonevec{p} = \sum_j p_j$ subject to $A'\vec{p} = \vec{b}$ and $\vec{p} \geq 0$.

Next, we need the trace-preserving condition. Provided $\chan$ is CPTP, if one part of the decomposition is trace-preserving, then the other will be as well, so we only need enforce the constraint on one of $\rho_+$ or $\rho_-$. Assume that we check $\rho_+$. The condition for a Choi state $\Phi^{AB} = \chan^A \otimes \id^B (\ket{\Omega} \bra{\Omega}^{AB})$ to be trace-preserving is:
\begin{equation}
\Tr_A (\Phi^{AB}) = \frac{\id}{d},
\label{eq:TPcriterion}
\end{equation}
where $d$ is the dimension of the subsystem. We need to convert this to a constraint on the vector $\vec{b}_+$ corresponding to $\phi_+$, which is given by $\vec{b}_+ = A \vec{p}_+$. First, note that all Paulis are traceless except for the identity $P_0 = \id$, so for the maximally mixed state:
\begin{equation}
\langle P_j \rangle = \Tr\qty(P_j \frac{\id}{d}) = \frac{\Tr(P_j)}{d} = \delta_{j,0},
\label{eq:expCondition}
\end{equation}
so if the first entry in a Pauli vector is always $\langle \id \rangle$, the maximally mixed state has Pauli vector:
\begin{equation}
\vec{b}_B = \begin{pmatrix} 1 \\ \vec{0}
\end{pmatrix}.
\end{equation}
where $\vec{0}$ is the zero vector. However, we need this to hold just for the reduced state on $B$ rather than the full Pauli vector. Consider that if the whole state is written 
$
\Phi^{AB} = \sum_{j,k} r_{j,k} P_j \otimes P_k
$ 
for some set of coefficients $r_{j,k}$, then the expectation values are given by:
\begin{equation}
\langle P_l \otimes P_m \rangle = \sum_{j,k} r_{j,k} \Tr(P_l P_j \otimes P_m P_k) = \sum_{j,k} r_{j,k} d^2 \delta_{j,l} \delta_{m,k} = d^2 r_{l,m}.
\end{equation}
The reduced state is:
\begin{equation}
\Tr_A (\Phi^{AB}) = \sum_{j,k} r_{j,k} \Tr_A[P_j \otimes P_k] = \sum_{j,k} r_{j,k} d \delta_{j,0} P_k = d \sum_k r_{0,k} P_k.
\end{equation}
and the entries of the reduced Pauli vector will be:
\begin{equation}
\langle P_m \rangle = d \sum_k r_{0,k} \Tr{P_m P_k} = d^2 r_{0,m} = \langle P_0 \otimes P_m\rangle.
\label{eq:reducedExp}
\end{equation}
So for condition \eqref{eq:TPcriterion} to hold for the reduced state on $B$, we combine equations \eqref{eq:expCondition} and \eqref{eq:reducedExp} to get:
\begin{equation}
\langle P_m \rangle = \langle P_0 \otimes P_m\rangle = \delta_{m,0}.
\end{equation}
That is, we just need to look at the entries of $\vec{b}_+$ corresponding to Paulis of the form $\id \otimes P_j$. These should all be zero except the first entry, which corresponds to $\langle \id \otimes \id \rangle$. Note that $\vec{b}_+ = A \vec{p}_+$ will in general not be normalised, but this does not matter, since we are only interested in whether or not entries are zero. We can use a binary matrix $M$ to pick out the values of interest. As an example we consider the two-qubit case, and assume that the entries are ordered as:
\begin{equation}
\vec{b}_+ = \begin{pmatrix}
\langle{\id \otimes \id}\rangle\\
\langle{\id \otimes X}\rangle\\
\langle{\id \otimes Y}\rangle\\
\langle{\id \otimes Z}\rangle\\
\langle{X \otimes \id}\rangle\\
\vdots \\
\langle{Z \otimes Z}\rangle
\end{pmatrix}.
\end{equation}
Here, we are only interested in the 2nd, 3rd and 4th entries. We form a new vector $\vec{c}$ by left multiplying with $M$:
\begin{equation}
\vec{c} = M \vec{b}_+ = \begin{pmatrix}
0 & 1 & 0 & 0 & 0 &\cdots& 0 \\
0 & 0 & 1 & 0 & 0 &\cdots& 0 \\
0 & 0 & 0 & 1 & 0 &\cdots& 0 \\
\end{pmatrix}\vec{b}_+ = \begin{pmatrix}
\langle{\id \otimes X}\rangle\\
\langle{\id \otimes Y}\rangle\\
\langle{\id \otimes Z}\rangle\\
\end{pmatrix}.
\end{equation}
Then the condition we need is just $\vec{c} = 0 $. To convert this to a condition on the $2N_S$-entry variable $\vec{p} = \begin{pmatrix} \vec{p}_+ \\ \vec{p}_- \end{pmatrix}$, we first pad $A$ with zeroes: $A_+ = \begin{pmatrix} A & \overline{0} \end{pmatrix}$, where $\overline{0}$ is the $(N_P \times N_S)$ zero matrix. We then have:
\begin{equation}
\vec{b}_+ = A \vec{p}_+ = A\vec{p}_+ + \overline{0} \vec{p}_- = \begin{pmatrix} A & \overline{0} \end{pmatrix} \begin{pmatrix} \vec{p}_+ \\ \vec{p}_- \end{pmatrix} = A_+ \vec{p},
\end{equation}
so that $
\vec{c} = M\vec{b}_+ = M A_+ \vec{p}
$. 
Therefore, our condition for trace-preserving $\rho_+$ is 
$
 M A_+ \vec{p} = 0
$. We can therefore specify the new optimisation problem as:
\begin{align*}
\text{\bf minimise}&\quad \lonevec{p} = \sum_j p_j\\
\text{\bf subject to}& \quad A' \vec{p} = \vec{b},\\
					& \quad \vec{p} \geq 0, \\
                    & \quad M A_+ \vec{p} = 0
\end{align*}
where $
A' = \begin{pmatrix} A & -A\end{pmatrix}$, and $A_+ = \begin{pmatrix} A & \overline{0}\end{pmatrix}
$, with $A$ and $\vec{b}$ having the same definitions as previously, $\overline{0}$ is the zero matrix with dimension the same as $A$, and with $M$ being the binary matrix that picks out the $\langle \id \otimes P_j \rangle$ entries from the vector $A_+ \vec{p}$. Most of this is straightforward to implement. The step that requires some care is in correctly constructing the matrix $M$, as it will depend on the choice of ordering of Pauli operators in the construction of $A$ and $\vec{b}$. If the $B$ subsystem has $n$ qubits, then we will need to constrain $4^n - 1$ non-trivial $\langle \id \otimes P_j \rangle$ expectation values to zero, so $M$ should have dimension $((4^n - 1) \times N_P)$. If the Paulis are ordered as in the example given above for $2$-qubit Choi states, then the construction is just 
$
M = \begin{pmatrix} \vec{0} & \id' & \vec{0} & \cdots & \vec{0} \end{pmatrix}
$,
 where $\id'$ is the $((4^n - 1) \times (4^n - 1))$ identity, and $\vec{0}$ denotes a column of zeroes. We have implemented this linear program in MATLAB, using the convex optimisation package CVX \cite{CVX}, and have made the code available from the repository Ref. \cite{channel_repo}.

\section{Properties of magic capacity \label{app:capacity_properties}}

\noindent \textbf{Faithfulness:} 
For any $n$-qubit stabiliser-preserving CPTP channel $\Lambda$, if $\rho \in \STABmth{2n}$ is a stabiliser state, then $(\Lambda \otimes \idn{n}) \rho$ is also a stabiliser state. So by the faithfulness of robustness of magic, $\robmag{(\Lambda \otimes \idn{n}) \rho}{} = 1$ for any input stabiliser state $\rho \in \STABmth{2n}$, and $\Mcap{\Lambda} = 1$. 

Suppose instead that $\chan$ is non-stabiliser-preserving, but still CPTP. Then there exists at least one stabiliser state $\rho \in \STABmth{2n}$ such that $(\chan \otimes \id) \rho$ is a normalised state, but not a stabiliser state. Then by faithfulness of $\mathcal{R}$ when applied to states, $\robmag{(\chan \otimes \id) \rho}{} > 1$, and so $\Mcap{\chan} > 1$.

\noindent\textbf{Convexity:}
Suppose we have a real linear combination of $n$-qubit CPTP maps $\chan_k$: 
\begin{equation}
\chan = \sum_k q_k \chan_k.
\end{equation}
There exists some optimal stabiliser state $\rho_*$  that achieves $\Mcap{\chan} = \robmag{\chan \otimes \id \qty(\rho_*)}{}$. Then
\begin{align}
\robmag{\qty(\chan \otimes \idn{n}) \rho_*}{} &  = \robmag{\sum_k q_k \qty[\qty(\chan_k \otimes \idn{n}) \rho_*]}{} \\
& \leq \sum_k \abs{q_k} \robmag{\qty(\chan_k \otimes \idn{n}) \rho_*}{}, 
\end{align}
where the last line follows by convexity of the robustness of magic. But each robustness  $\robmag{\qty(\chan_k \otimes \idn{n}) \rho_*}{} $ can be no larger than $\Mcap{\chan_k}$. So we have:
\begin{equation}
    \Mcap{ \sum_k q_k \chan_k} \leq \sum_k \abs{q_k} \Mcap{\chan_k}.
\end{equation}

\section{\label{app:diagonal}Calculating monotones for diagonal channels}

\subsection{Reducing the problem size}
As mentioned earlier, the size of the optimisation problem for calculating our monotones (as well as $\RChoi{\chan}$) quickly becomes prohibitively large for $n$-qubit states, since the number of stabiliser states increases super-exponentially with $n$ (Table \ref{tab:stabstatenums}). 
\begin{table}[htbp]
\centering
\begin{tabular}{c|c}
$n$ & $N_S$ \\
\hline
1 & 6\\
2&60\\
3 & 1,080 \\
4 & 36,720 \\
5 & 2,423,520 \\
6 & 315,057,600
\end{tabular}
\caption{Number of pure stabiliser states $N_S$ for number of qubits $n$.}
\label{tab:stabstatenums}
\end{table}
The issue is even worse than it first appears, since for an $n$-qubit channel we must in general consider $2n$-qubit stabiliser states. Direct calculation of either monotone is impractical for $n$-qubit channels with $n>2$. This difficulty is aggravated when calculating the capacity as in principle we have to repeat the optimisation for every $\qty(\chan \otimes \idn{n})\op{\phi}$ such that $\ket{\phi} \in \STABmth{2n}$. 
In some cases we can ameliorate these problems by looking for Clifford gates that commute with the channel of interest. Here we consider the case where $\chan$ is a diagonal channel, meaning it has a Kraus representation where each Kraus operator is diagonal in the computational basis. This of course includes diagonal unitaries as a special case. One could likely reduce the problem size further by exploiting symmetries of channels using techniques similar to those used in Ref. \cite{Heinrich2018}, but we will not consider this strategy here.

It is straightforward to see how the problem can be simplified for calculating $\RChoi{\chan}$ and $\mathcal{R}_*(\chan)$. If $\chan$ is diagonal, the operation $\chan \otimes \idn{n}$ commutes with any sequence of CNOTs targeted on the last $n$ qubits. But the maximally entangled state $\ket{\Omega_n}$ can be written:
\begin{equation}
    \ket{\Omega_n} = U_C (\ket{+}^{\otimes n} \otimes \ket{0}^{\otimes n}).
\end{equation}
Here $U_C = \otimes_{j=1}^n U_j$, where $U_j$ is the CNOT controlled on qubit $j$ and targeted on qubit $n+j$. By the monotonicity of robustness of magic, we immediately see that:
\begin{equation}
    \RChoi{\chan} = \mathcal{R}\qty[\outstate{\chan}{\op{\Omega_n}}{n}] = \mathcal{R} \qty[\chan\qty(\op{+}^{\otimes n})\otimes \op{0}^{\otimes n}] = \mathcal{R} \qty[\chan\qty(\op{+}^{\otimes n})].
\end{equation}

For the channel robustness we would like to decompose $\chan\qty(\op{+}^{\otimes n})$ in terms of states $\rho_\pm \in \STABmth{n}$, but need to take care that the trace condition $\Tr_A(\rho'_\pm) = \idn{n}/2^n$ is satisfied for the equivalent $2n$-qubit Choi states $\rho'_\pm$. In Appendix \ref{app:tracecond} we show that the criterion is satisfied provided all diagonal elements of $\rho_\pm$ are equal to $1/2^n$. So for diagonal channels we can write:
\begin{equation}
    \mathcal{R}_*(\chan) = \min_{ \rho_\pm \in \STABmth{n}} \qty{1 + 2p : (1+p) \rho_+ - p \rho_- = \chan(\op{+}^{\otimes n}), \, p \geq 0, \bra{x}\rho_\pm \ket{x} = \frac{1}{2^n},  \forall x } .
\end{equation}
So calculation of $R_*(\chan)$ and $\mathcal{R}(\Phi_\chan)$ is tractable up to five qubits provided $\chan$ is diagonal. We will see below in Section \ref{app:cap_affine} that this is also true for the magic capacity.

\subsection{Trace condition for diagonal channels \label{app:tracecond}}
Consider that the Choi state for a diagonal channel has a decomposition 
\begin{equation}
    \Phi_\chan = U_C \qty(\chan(\op{+}^{\otimes n}) \otimes \op{0}^{\otimes n})U_C^\dagger= (1 + p) \rho_+ - p \rho_-,
    \label{eq:plus_zero_decomp}
\end{equation}
where $U_C = \otimes_{j=1}^n U_j$ is the tensor product of CNOTs $U_j$ that are controlled on the $j$th qubit and targeted on the  $n+j$th. Then 
\begin{equation}
  \chan(\op{+}^{\otimes n}) \otimes \op{0}^{\otimes n}= (1 + p) \rho'_+ - p \rho'_-,\label{eq:E_plus_tensor_zero_decomp}
\end{equation}
where $\rho'_\pm$ are still stabiliser states since $U_C$ is Clifford. Now consider the stabiliser-preserving channel $\idn{n} \otimes \Lambda$ that resets the last $n$ qubits to $\op{0}^{\otimes n}$. Applying this to both sides of equation \eqref{eq:E_plus_tensor_zero_decomp} we get a new decomposition
\begin{equation}
\chan(\op{+}^{\otimes n}) \otimes \op{0}^{\otimes n}= (1 + p) \rho''_+ \otimes \op{0}^{\otimes n}  - p \rho''_-\otimes \op{0}^{\otimes n}.
\end{equation}
Then referring back to equation \eqref{eq:plus_zero_decomp}, we obtain $\rho_\pm = U_C \qty(\rho''_+ \otimes \op{0}^{\otimes n})U_C^\dagger$. So, the trace-preserving condition becomes:
\begin{align}
    \frac{\idn{n}}{2^n} = \Tr_A \qty(\rho_\pm) &= \Tr_A \qty(U_C \qty(\rho''_+ \otimes \op{0}^{\otimes n})U_C^\dagger)  \\
                & = \sum_x \bra{x}^A U_C \qty(\rho''_+ \otimes \op{0}^{\otimes n})U_C^\dagger \ket{x}^A,
\end{align}
where $\ket{x}$ are the computational basis states on subsystem $A$. Recalling that $U_C$ can be written as a tensor product of CNOTs $U_C = \otimes_{j=1}^n U_j$ one can check that this equation can be written:
\begin{equation}
    \frac{\idn{n}}{2^n} = \sum_x \bra{x} \rho''_\pm \ket{x} \op{x}.
\end{equation}
Therefore, the decomposition corresponds to a pair of trace-preserving channels provided that all diagonal elements of $\rho''_\pm$ are equal to $1/2^n$. 

For a given diagonal channel, there always exists a decomposition that satisfies these conditions and has  $\ell_1$-norm equal to the channel robustness as defined for the full Choi state. We do not give the full proof here, but sketch the argument. Given any diagonal channel $\chan$ decomposition of the full Choi state $\Phi_\chan = (1 + p) \rho_+ - p \rho_-$ satisfying the trace condition, one can always find a new decomposition $\Phi_\chan = (1 + p) \Lambda(\rho_+) - p \Lambda(\rho_-)$ where $\Lambda(\rho_\pm)$ still satisfy $\Tr_A(\Lambda(\rho_\pm))$, but are now the Choi states for diagonal channels. The map $\Lambda$ used to obtain this decomposition is in effect an error correction circuit that takes general stabiliser Choi states to the subspace corresponding to the diagonal channels. Specifically, we note that the Choi states for diagonal maps $\curlyT$ have the general form:
\begin{equation}
   \Phi_\curlyT = \frac{1}{2^n}\sum_{j,k} c_{j,k} \ket{j}^A\ket{j}^B\bra{k}^A\bra{k}^B.
   \label{eq:thischoistateisdiagonal}
\end{equation}
In general $c_{j,k}$ can be complex or zero, but terms on the diagonal are constrained. In particular, trace-preserving diagonal channels cannot change the weight of particular computational basis states, so the probability distribution for computational basis states will be the same as for $\ket{\Omega}$:
\begin{equation}
    \bra{p,q}  \Phi_\curlyT \ket{p,q} = \frac{1}{2^n} \delta_{p,q}.
\end{equation}
The circuit $\Lambda$ is defined by the following steps. For each $j$ from $1$ to $n$:
\begin{enumerate}
    \item Perform a parity measurement ($Z\otimes Z$) between qubits $j$ and $n + j$.
    \item If even parity (+1 outcome), do nothing. If odd parity (-1 outcome), perform an $X$ gate on qubit $j$.
\end{enumerate}
This stabiliser-preserving channel leaves Choi states for diagonal maps (and crucially, the target Choi state $\Phi_\chan$) invariant, but updates general Choi states to have the form \eqref{eq:thischoistateisdiagonal}. One can check that the circuit preserves the property $\Tr_A(\rho_\pm) = \idn{n}/2^n$. We then obtain a decomposition in the desired form:
\begin{equation}
    \Lambda(\Phi_\chan) = \Phi_\chan = (1 + p) \Lambda(\rho_+) - p \Lambda(\rho_-)
\end{equation}
Where $\Lambda(\rho_\pm) = (\chan_\pm \otimes \id)\op{\Omega} $ are Choi states for $n$-qubit diagonal channels $\chan_\pm$. But as described above, the CNOT sequence $U_C$ commutes with diagonal channels acting on the first $n$ qubits, so we can obtain $n$-qubit representatives of these channels:
\begin{equation}
    \chan_\pm(\op{+}^{\otimes n}) \otimes \op{0}^{\otimes n} = U_C \qty((\chan_\pm \otimes \id)\op{\Omega}) U_C^\dagger.
\end{equation}
Discarding the last $n$ qubits we obtain the desired $n$-qubit decomposition:
\begin{equation}
    \chan(\op{+}^{\otimes n}) = (1 + p) \chan_+(\op{+}^{\otimes n}) - p \chan_-(\op{+}^{\otimes n}).
\end{equation}

\subsection{Magic capacity in the affine space picture\label{app:cap_affine}}
In this section we will make use of the formalism due to Dehaene and De Moor, in which stabiliser states are cast in terms of affine spaces and quadratic forms over binary vectors \cite{Dehaene2003,Gross2008}, to prove the following theorem:
\begin{theorem}[Capacity for diagonal operations]\label{thm:cap_affine}
Suppose the $n$-qubit channel $\chan_D$ is diagonal. Let:
\begin{equation}
    \ket{\mathcal{K}} = \frac{1}{|\mathcal{K}|^{1/2}} \sum_{x \in \mathcal{K}} \ket{x},
    \label{eq:K_states}
\end{equation}
where $x \in \mathbb{F}_2^n$ are binary vectors and $\mathcal{K} \subseteq \mathbb{F}_2^n $ is an affine space. Then:
\begin{equation}
\Mcap{\chan_D} = \max_{\mathcal{K}} \robmag{\chan_D\qty(\op{\mathcal{K}})}{}.
\end{equation}
\end{theorem}
That is, given an $n$-qubit channel $\chan$, provided the channel is diagonal, the capacity $\Mcap{\chan}$ may be calculated by optimisation over only the $n$-qubit states $\ket{\mathcal{K}}$ as defined in equation \eqref{eq:K_states}, rather than over all $2n$-qubit stabiliser states.

We first review the formalism of Ref. \cite{Dehaene2003}. Computational basis states $\ket{x}$ can be labelled by binary column vectors $x = \qty(x_1,\ldots,x_n)^T \in \mathbb{F}_2^n$, so that $x_j \in \qty{0,1}$ relates to the $j$th qubit. Any pure $n$-qubit stabiliser state may be written:
\begin{equation}
\ket{\mathcal{K},q,d} = \frac{1}{|\mathcal{K}|^{1/2}} \sum_{x \in \mathcal{K}} i^{d^T x} (-1)^{q(x)} \ket{x},
\end{equation}
where $\mathcal{K} \subseteq  \mathbb{F}_2^n$ is an affine space, $d$ is some fixed binary vector, and $q(x)$ has the form:
\begin{equation}
q(x) = x^T Q x + \lambda^T x.
\end{equation}
Here $Q$ is a binary, strictly upper triangular matrix, $\lambda$ is a vector, and addition is modulo 2. Conversely, any state that can be written in this way is a stabiliser state. 

An affine space $\mathcal{K}$ is a linear subspace $\mathcal{L}$ shifted by some constant binary vector $h$, modulo 2: $\mathcal{K} = \mathcal{L} + h$. Every affine space is related in this way to exactly one linear subspace, and the dimension $k = \dim(\mathcal{K})$ of an affine space means the dimension of the corresponding subspace. Instead of enumerating all elements of an affine space, we can specify it by a shift vector $h$ and an $n \times k$ matrix where each column is one of the generators of the corresponding linear space:

\begin{equation}G = 
\begin{pmatrix}
\vec{g}_1 & \vec{g}_2  & \cdots & \vec{g}_k
\end{pmatrix} = \begin{pmatrix}
g_{1,1} & g_{1,2}  & \cdots & g_{1,k} \\
\vdots &  \vdots & &\vdots\\
g_{j,1} & g_{j,2}  & \cdots & g_{j,k} \\
\vdots &  \vdots & &\vdots\\
g_{n,1} & g_{n,2}  & \cdots & g_{n,k} \\
\end{pmatrix} .
\end{equation}
We have freedom in our choice of $k$ independent generators, and we can transform between equivalent generating sets by adding any two columns of $G$. We are also free to swap any two columns. A general transform between generating sets can therefore be represented by an invertible matrix $S$ of dimension $k \times k$, multiplying on the right $G \longrightarrow GS$.

Any non-trivial linear transformation of the affine space can be fully specified by the transformation of the generators and the shift vector. In particular, we can represent the action of a single CNOT by multiplication on the left by a matrix $C$. If the CNOT has control qubit $j$ and target qubit $k$, then $C$ has 1s on the diagonal, a 1 in the $j$th element of the $k$th row, and zeroes everywhere else. A sequence for a $2n$-qubit system, in which CNOTs are always controlled on the first $n$ qubits, and targeted on the last $n$ qubits can be represented in block form:
\begin{equation}
C = \begin{pmatrix}
\id & 0 \\
M & \id \\
\end{pmatrix},
\label{eq:CNOTmatrix}
\end{equation}
where each block has dimension $ n \times n $, and $M$ can be any binary matrix. We use this formalism to prove the following lemma, which leads directly to Theorem \ref{thm:cap_affine}:
\begin{lemma}[Equivalences for diagonal channels]\label{thm:diagonal}
Suppose $\chan_D$ is a diagonal CPTP channel. Then:
\begin{enumerate}
\item All input stabiliser states with the same affine space $\mathcal{K}$ result in the same final robustness:
\begin{equation}
\robmag{(\chan_D \otimes \id)\op{\mathcal{K},q,d}}{} = \robmag{(\chan_D \otimes \id)\op{\mathcal{K},q',d'}}{}, \quad \forall q,q',d,d'.
\end{equation}
\item Given a $2n$-qubit state $ \ket{\phi}\in \STABmth{2n}$ , there exists some $n$-qubit $\ket{\phi'} \in \STABmth{n}$ such that:
\begin{equation}{\robmag{(\chan_D \otimes \idn{n})\op{\phi}}{} = \robmag{\chan_D \qty(\op{\phi'})}{}} .
\end{equation}
\end{enumerate}
\end{lemma}
\begin{proof}
We first prove statement 1. Since robustness of magic is invariant under Clifford unitaries, we need to show that there exists a Clifford unitary $U$ that converts $(\chan_D \otimes \id)\op{\mathcal{K},q,d}$  to $(\chan_D \otimes \id)\op{\mathcal{K},q',d'}$. A suitable choice for $U$ is one such that $U\ket{\phi_{\mathcal{K},q,d}} = \ket{\phi_{\mathcal{K},q',d'}}$, and, crucially, that commutes with the channel $\chan_D$. Since $\chan_D$ is given to be diagonal, any diagonal Clifford $U$ will suffice. The affine space $\mathcal{K}$ remains unchanged, so we only need show there is always a diagonal Clifford that maps $q \to q'$ and $d\to d'$ for any $q$, $q'$, $d$ and $d'$. That this is always possible is perhaps already evident from Ref. \cite{Dehaene2003}, but for completeness we give the argument here. 

We can convert $d$ to $d'$ using appropriately chosen $S_j$ gates, meaning the gate $\mathrm{diag}(1,i)$ acting on the $j$th qubit.  Consider the action of $S_j$ on a basis vector:
\begin{equation}
S_j\ket{x} = \begin{cases}
\ket{x} \, \text{if} \, x_j = 0 \\
i\ket{x} \, \text{if} \, x_j = 1
\end{cases}.
\end{equation}
If we define basis vector $e_j$ so that it has $1$ in the $j$th position and zeroes elsewhere, we can write the action of $S_j$ as:
\begin{equation}
S_j \ket{x} = i^{e_j^T x} \ket{x}.
\end{equation}
Note that the form of this equation is independent of the value of $x$, so we can write:
\begin{align}
S_j \ket{\phi_{\mathcal{K},q,d}} &= \frac{1}{|\mathcal{K}|^{1/2}} \sum_{x \in \mathcal{K}} i^{d^T x} (-1)^{q(x)} S_j \ket{x} \\
								& = \sum_{x \in \mathcal{K}} i^{(d^T + e_j^T) x} (-1)^{q(x)} S_j \ket{x}.
\end{align}
So, we can flip any bit of $d$ by applying the correct $S$ gate. The quadratic form $q(x)$ is left unchanged.

Now consider $q(x) = x^T Q x + \lambda^Tx$, which we must convert to some other ${q'(x) = x^T Q' x + \lambda'^T x}$. We can use the same trick as above to convert any $\lambda$ to any other $\lambda'$, by replacing $S_j$ with the $Z_j$ gate, i.e. $\mathrm{diag}(1,-1)$ acting on the $j$th qubit. For $Q$ we can use the controlled-$Z$ gate between the $j$th and $k$th qubit, which we denote $CZ_{jk}$. This has the following effect on a basis state:
\begin{equation}
CZ_{jk} \ket{x} = (-1)^{x^T M_{jk} x} \ket{x},
\end{equation}
where $M_{jk}$ is the $n\times n$ matrix with a $1$ in position $(j,k)$ and zeroes everywhere else. The set of all $\qty{M_{jk}}$ form a basis for $n \times n$ binary matrices, hence we can convert any $Q$ to any other $Q'$ by an appropriately chosen sequence of $CZ$ gates, leaving $d$ and $\lambda$ untouched. This completes the proof of statement 1.

Now to prove statement 2. From statement 1 any stabiliser state $\ket{\phi}$ is equivalent to:
\begin{equation}
\ket{\mathcal{K}} = \frac{1}{|\mathcal{K}|^{1/2}} \sum_{x \in \mathcal{K}} \ket{x},
\end{equation}
up to some diagonal Clifford, for some $\mathcal{K}$. The strategy is to find a Clifford unitary $U$ that commutes with $\chan_D$, and converts the $2n$-qubit stabiliser state $\ket{\mathcal{K}}$ to some product of two $n$-qubit states $\ketcal{K}{'} = \ketcal{K}{'_A} \otimes \ketcal{K}{'_B}$. Then we have:
\begin{align}
\mathcal{R}[(\chan_D \otimes \idn{n})\op{\mathcal{K}}] & = \mathcal{R}\qty[(\chan_D \otimes \idn{n})(\op{\mathcal{K}'_A} \otimes \op{\mathcal{K}'_B})] \\
									& = \mathcal{R}\qty[\chan_D(\op{\mathcal{K}'_A}) \otimes \op{\mathcal{K}'_B}]  = \mathcal{R}\qty[\chan_D\qty(\op{\mathcal{K}'_A})],
\end{align}
where the last step follows as $\ketcal{K}{'_B}$ is a stabiliser state so makes no contribution to the robustness. The final state $\ketcal{K}{'}$ can be factored as $\ketcal{K}{'_A} \otimes \ketcal{K}{'_B}$ provided its generator $G'$ can be written in block matrix form as:

\begin{equation}
G' = \begin{pmatrix}
G'_A & 0 \\
0 & G'_B \\
\end{pmatrix},
\label{eq:productgenerator}
\end{equation}
where $G'_A$ and $G'_B$ have $n$ rows, and represent the generators for affine spaces $\mathcal{K}'_A$ and $\mathcal{K}'_B$.

We now show that we can always reach this form by a Clifford $U_C$ comprised of a sequence of CNOTs targeted on the last $n$ qubits. Such a sequence always commutes with $\chan_D \otimes \idn{n}$. Suppose we have some $2n \times k$ generator $G$ for an affine space $\mathcal{K}$ with $k = \dim(\mathcal{K})$:
\begin{equation}
G = \begin{pmatrix}
G_A \\
G_B \\
\end{pmatrix},
\end{equation}
where $G_A$ and $G_B$ are each $n \times k $ submatrices. The full matrix $G$ will have rank $k$, and $G_A$ will have some rank $m\leq k$. Either $G_A$ is already full rank ($m=k$), or it can be reduced to the following form by elementary column operations, which is equivalent to multiplication on the right by a $k \times k$ matrix $S$:
\begin{equation}
G_A \longrightarrow G_A S = \begin{pmatrix} G_A' & 0 \end{pmatrix},
\end{equation}
where $G_A'$ is $n \times m$ (and hence full column rank), and $0$ is $n \times (k - m)$. Multiplying $G$ on the right by $S$, we interpret as a change in the choice of generating set:
\begin{equation}
G \longrightarrow  GS = \begin{pmatrix}
G_A S \\ G_B S \\
\end{pmatrix} = 
\begin{pmatrix}
G'_A & 0 \\ G''_B  & G'_B \\
\end{pmatrix}.
\end{equation}
Now, apply the Clifford $U_C$ described by the matrix $C$ in equation \eqref{eq:CNOTmatrix}. This transforms the generator to:
\begin{equation}
G' = CGS  = \begin{pmatrix}
\id & 0 \\
M & \id \\
\end{pmatrix} 
\begin{pmatrix}
G'_A & 0 \\ G''_B  & G'_B \\
\end{pmatrix} = \begin{pmatrix}
G'_A & 0 \\ M G'_A  +  G''_B  & G'_B \\
\end{pmatrix}.
\end{equation}
Note that if $G_A$ was already full rank, the change of generating set is not necessary. If we can set the bottom-left submatrix to zero, then $U_C\ketcal{K}{}$ can be factored as described above. This is possible if there exists a binary matrix $M$ such that 
$
 M G'_A  =  G''_B
$.
But $G'_A$ has full column rank $m$, so there exists an $m \times n$ left-inverse $G'^{-1}_{A,\mathrm{left}}$ such that
$
G'^{-1}_{A,\mathrm{left}} G'_A = \id,
$
where $\id$ is $m \times m$. Then we can set $M = G''_B G'^{-1}_{A,\mathrm{left}}$, so that:
\begin{equation}
 M G'_A  = G''_B G'^{-1}_{A,\mathrm{left}} G'_A = G'_B \id =  G''_B.
\end{equation}
Then $G' = CGS$ is in the form \eqref{eq:productgenerator}, so $U_C \ketcal{K}{} = \ketcal{K}{'_A} \otimes \ketcal{K}{'_B}$, as required.
\end{proof}
Lemma \ref{thm:diagonal} shows that if $\chan_D$ is diagonal then for any $2n$-qubit stabiliser state $\ket{\phi}$ we have that $\robmag{\outstate{\chan_D}{\op{\phi}}{n}}{} = \robmag{\chan_D\qty(\op{\mathcal{K}})}{}$ for some $n$-qubit affine space $\mathcal{K}$. This shows that the capacity can be calculated by maximising over just the representative states $\ketcal{K}{}$, proving Theorem \ref{thm:cap_affine}. Table \ref{tab:spaces} illustrates the reduction in problem size. For example, whereas naively for a two-qubit channel we would need to calculate robustness for all $36,720$ four-qubit stabiliser states, using the result above we only need check one stabiliser state for each of the $7$ non-trivial affine spaces. Cases up to five qubits are now tractable using this method.

\begin{table}[htbp]
\centering
\begin{tabular}{c|c|c|c}
$n$ &  Stabiliser states & Total affine spaces & Non-trivial affine spaces\\
\hline
2 & 60 & 11 & 7\\
3 & 1,080 & 51 & 43\\
4 & 36,720 & 307 & 291\\
5 & 2,423,520 & 2451 & 2419\\
\end{tabular}
\caption{\label{tab:spaces} Number of $n$-qubit stabiliser states compared with number of affine spaces. By trivial affine spaces we mean those comprised of a single element, which correspond to computational basis states. Diagonal CPTP channels act as the identity on such states.}
\end{table}

\subsection{Dimension of affine space\label{app:dim_affine}}
Here we make further observations that will help interpret numerical results from Section \ref{sec:numerical} of the main text.
\begin{obs}[Dimension of affine space limits achievable robustness]\label{obs:affineDimension}
Suppose $U$ is a diagonal unitary acting on $n$ qubits, and suppose $\ket{\mathcal{K}}$ is a stabiliser state associated with some affine space $\mathcal{K}$, $k= \dim(\mathcal{K})$. Then $\mathcal{R}(U\ket{\mathcal{K}}) = \mathcal{R}(U' \ket{\phi'})$ where $U' \ket{\phi'}$ is a state on only $k$ qubits, and $U'$ is some $k$-qubit unitary. Therefore $\mathcal{R}(U\ket{\mathcal{K}})$ is upper-bounded by the maximum robustness achievable for a $k$-qubit state.
\end{obs}
\begin{proof}
We prove the result by showing that there is a sequence of Clifford gates that takes $U\ket{\mathcal{K}}$ to the product of a $k$-qubit state and an $(n-k)$-qubit stabiliser state. We know from Lemma \ref{thm:diagonal} that for diagonal unitaries, all states with same affine space result in the same robustness, so it is enough to consider the state:
\begin{equation}
\ket{\mathcal{K}} = \frac{1}{\sqrt{\abs{\mathcal{K}}}} \sum_{x \in \mathcal{K}} \ket{x}.
\end{equation}
A diagonal unitary will map this to:
\begin{equation}
U \ket{\mathcal{K}} = \frac{1}{\sqrt{\abs{\mathcal{K}}}} \sum_{x \in \mathcal{K}} e^{i \theta_x}\ket{x},
\end{equation}
where $\qty{e^{i \theta_x}}$ will be some subset of the diagonal elements of $U$. The affine space $\mathcal{K}$ will have a generator matrix of rank $k$. As we saw in Lemma \ref{thm:diagonal}, a sequence of elementary row operations on the generator matrix can be realised by a sequence of CNOT gates. So we can use Clifford gates to transform any rank $k$ generator matrix as:
\begin{equation}
G \longrightarrow G' = A G = \begin{pmatrix}
\id \\
0 
\end{pmatrix},
\end{equation}
where $\id$ is the $k \times k$ identity. Each element of $\mathcal{K}$ can be written $x = \sum_j g_j + h$, where $\sum_j g_j$ is some combination of columns of $G$, and $h$ is a fixed shift vector. The transformation $A$ corresponds to a sequence of CNOTs that we collect in a single Clifford unitary $U_A$, that acts on $n$-qubit computational basis states $\ket{x}$, where $x \in \mathcal{K}$, as follows:
\begin{equation}
U_A \ket{x} = \ket{y(x)} \otimes \ket{h'},
\end{equation}
where $h'$ is an $(n-k)$-length vector, and $y(x)$ is a $k$-length vector given by:
\begin{equation}
\begin{pmatrix}
y(x) \\
h'
\end{pmatrix} = A x = \sum_j Ag_j +Ah.
\end{equation}
Note that $y(x)$ is only defined for $x \in \mathcal{K}$, and that $h'$ is independent of $x$. Elements $x\in \mathbb{F}_2^n$ that are not in $\mathcal{K}$ could be mapped to a vector where the last $n-k$ bits are not $h'$, but these never appear as terms of $U \ket{\mathcal{K}}$. Since $U_A$ must preserve orthogonality, each $\ket{x}$,  where $x \in \mathcal{K}$, maps to a distinct element of the $k$-qubit basis set $\qty{\ket{y}}$. In fact, since $y$ are length $k$ and there are $2^k$ distinct elements, they must form the $k$-bit linear space $\mathcal{L}' = \mathbb{F}_2^k$. So we can write:
\begin{align}
U_A U \ketcal{K}{} &= \frac{1}{\sqrt{\abs{\mathcal{L}'}}} \sum_{y \in \mathcal{L}'} e^{i \theta'_y} \ket{y} \otimes \ket{h'} \\
		& = (U' \ket{\mathcal{L'}}) \otimes \ket{h'},
\end{align}
where $\ket{\mathcal{L}'}$ is a $k$-qubit stabiliser state, and $U'$ is the $k$-qubit diagonal unitary with $e^{i\theta'_{y(x)}} = e^{i\theta_x}$ as the non-zero elements. The state $\ket{h'}$ is a stabiliser state, so cannot contribute to the robustness of $U_A U \ketcal{K}{}$, and therefore $\mathcal{R}(U\ket{\LofK{}}) = \mathcal{R}(U_A U \ket{\LofK{}}) = \mathcal{R}(U' \ket{\mathcal{L}'})$, where $U' \ket{\mathcal{L}'}$ is a $k$-qubit state.
\end{proof}
Recall that in Section \ref{sec:numerical} of the main text, we found that highly structured examples of diagonal unitaries $U$ exist where $\mathcal{C}(U)$ is strictly larger than $\mathcal{R}(\Phi_U)$, whereas for all the random diagonal unitaries sampled, we found them to be exactly equal. We can now explain this by a concentration effect, in conjunction with Observation \ref{obs:affineDimension}. The $n$-qubit random diagonal gates concentrate (with high probability) within a narrow range of values for the magic capacity, close to the maximum possible magic capacity for an $n$-qubit diagonal gate. If $\mathcal{R}(\Phi_U) < C(U)$ then by Theorem ~\ref{thm:cap_affine} we must have that $C(U)=\mathcal{R}(U \op{\mathcal{K}}U^{\dagger}  )$ for some affine space $\mathcal{K}$ of non-maximal dimension.   However, $U \op{\mathcal{K}}U^{\dagger} $ is Clifford equivalent to an $(n-1)$-qubit stabiliser state acted on by a diagonal unitary.  Then $\mathcal{R}(U \op{\mathcal{K}}U^{\dagger}  )$ would be upper bounded by the maximum $C(\mathcal{E})$ for $(n-1)$-qubit diagonal unitaries.  But if $C(\mathcal{E})$ is close to the maximum possible for $n$-qubit diagonal unitaries, then it is impossible for $U \op{\mathcal{K}}U^{\dagger} $ to achieve the magic capacity.

Finally, we consider the special case of multi-control phase gates $M_{t,n}$, which we defined in the main text as:
\begin{equation}
M_{t,n} = \mathrm{diag} ( \exp(  i \pi / 2^t ) , 1 , 1 , \ldots , 1 ), \quad t \in \mathbb{Z}.
 \label{eq:SuppCCphase}
\end{equation} Note that the gate $M_{t,n}$ acts as the identity on states $\ketcal{K}{}$ unless $\mathcal{K}$ contains the zero vector $0^n = (0,\ldots,0)^T$, so if $0^n \notin \mathcal{K}$, we get $\mathcal{R}(M_{t,n} \ket{\mathcal{K}}) = 1$. But if $0^n\in \mathcal{K}$, then $\mathcal{K}$ is a linear subspace. So for this type of gate, to find all possible values of $\mathcal{R}(M_{t,n} \ket{\mathcal{K}})>1$ we need only consider linear subspaces. The following theorem implies that we actually only need solve one optimisation for each possible \emph{dimension} of linear subspace rather than one for every linear subspace.
\begin{theorem}\label{thm:CtEquiv}
Consider the $n$-qubit gate $M_{t,n}$ defined by equation \eqref{eq:SuppCCphase}, and let $\mathcal{L}_A$ and $\mathcal{L}_B$ be linear subspaces such that $\dim(\mathcal{L}_A) = \dim(\mathcal{L}_B) = k$. Then:
\begin{equation}
\robmag{M_{t,n} \ket{\mathcal{L}_A}}{} = \robmag{M_{t,n} \ket{\mathcal{L}_B}}{}.\label{eq:CtRobEquiv}
\end{equation}
\end{theorem}
\begin{proof}

We largely repeat the arguments of Observation \ref{obs:affineDimension}, for the special case where the phases are given by:
\begin{equation}
\theta_x = \begin{cases}
\pi/2^t \quad & \text{if}\,  x = \vec{0} \\
0 \quad & \text{otherwise}
\end{cases}
\label{eq:phases}
\end{equation}
Since $ \dim(\mathcal{L}_A) = \dim(\mathcal{L}_B)$, their generator matrices $G_A$ and $G_B$ have the same rank. It follows from the arguments of Observation \ref{obs:affineDimension} that there exists an invertible $C$, corresponding to a sequence of CNOT gates, such that $G_B = C G_A$, and $\ketcal{L}{_A} = U_C \ketcal{L}{_A}$, where $U_C$ is a unitary Clifford operation.

If we consider instead the state $M_{t,n} \ketcal{L}{_A}$, which involves terms in the same basis vectors as $\ketcal{L}{_A}$, we just need to track what happens to the phase $\exp(i \theta_0)$. Clearly, since any CNOT acts as the identity on $\ket{0^n}$, we obtain:
\begin{equation}
     U_C  M_{t,n} \ketcal{L}{_A}  = \frac{1}{2^{k/2}} \sum_{x \in \mathcal{L}_B} \exp(i \theta_x) \ket{x} = M_{t,n} \ketcal{L}{_B}
\end{equation}
Since $U_C$ is a reversible Clifford operation, by monotonicity of robustness of magic, equation \eqref{eq:CtRobEquiv} follows.
\end{proof}
From Theorem \ref{thm:CtEquiv}, then, to find $\Mcap{M_{t,n}}$, we only need calculate $\robmag{M_{t,n}\ketcal{L}{}}{}$ for a single representative subspace for each possible value of $\dim(\mathcal{L})$. Recall that for $n$-qubit stabiliser states $\ketcal{L}{}$, $k=\dim{\mathcal{L}}$ can take integer values from $0$ to $n$. The states with $k=0$ correspond to single computational basis states without superposition, so are unaffected by phase gates. That is, for $n$-qubit multicontrol phase gates we only have to calculate $n$ robustnesses. Compare this to the number of optimisation problems we would need to solve without using the above observations (Table \ref{tab:spaces}).

We can go further. From Observation \ref{obs:affineDimension} we know that for a subspace with $\dim(\mathcal{L})=k < n $, it must be the case that $M_{t_n} \ketcal{L}{}$ is Clifford-equivalent to $(U' \ket{\mathcal{L'}})\otimes \ket{h'} $ for the $k$-qubit state  $\ket{\mathcal{L'}}$ and $(n-k)$-qubit computational basis state $\ket{h'}$, and some diagonal $k$-qubit unitary $U'$. By inspection of the phases given by equation \eqref{eq:phases}, $U'$ can only be the $k$-qubit multicontrol gate $M_{t,k}$. This leads to the following statement:

\begin{obs}[$n$-qubit multicontrol gates]\label{obs:M_t_n_dimension}
For any fixed $t$ and $n$-qubit state $\ketcal{L}{}$ where $\dim(\mathcal{L}) = k < n$, we have:
\begin{equation}
    \robmag{M_{t,n}\ketcal{L}{}}{} = \robmag{M_{t,k}\ketcal{L}{'}}{}
\end{equation}
where $\ketcal{L}{'}$ is the $k$-qubit state with $\mathcal{L}' = \mathbb{F}_2^k$.
\end{obs}

\begin{table}[thbp]
\centering
\begin{tabular}{|c|c|c|c|c|}
\hline
Linear subspace & \multicolumn{4}{|c|}{Number of qubits, $n$} \\
\cline{2-5}
dimension, $k$ & 2 & 3 & 4 & 5\\
\hline
1 & $1.414$ & $1.414$ & $1.414$ & $1.414$ \\
2 & $\color{red}{\mathbf{1.849}}$ & $1.849$ & $1.849$ & $1.849$ \\
3 & - & $\color{red}{\mathbf{2.195}}$ & $2.195$ & $2.195$ \\
4 & - & - & $\color{red}{\mathbf{2.264}}$ & $\color{red}{\mathbf{2.264}}$ \\
5 & - & - & - & $2.195$ \\
\hline
\end{tabular}
\caption{Final robustness after multicontrol-$T$ gate applied to input stabiliser states $\ketcal{L}{}$ with $k = \dim(\mathcal{L})$. In each column, the maximum robustness (i.e. the capacity) is highlighted red.}
\label{tab:affineEquality}
\end{table}
Observation \ref{obs:M_t_n_dimension} partially justifies our Conjecture \ref{conj:fourqubitMulticontrol} in Section \ref{sec:numerical} of the main text, that for fixed $t$, the maximum increase in robustness achievable for $M_{t,n}$, over any $n$, is given by $\robmag{M_{t,K}\ket{+}^{\otimes K}}{}$, for some finite number of qubits $K$. To unpack this claim further, let us consider the maximisation over input stabiliser states performed to calculate the capacity $\mathcal{C}$. In this Appendix, we have seen that for the family of gates $M_{t,n}$, we only need to calculate robustness for one representative input stabiliser state for each possible \emph{dimension} of linear subspace; that is,  for $M_{t,n}$ there are only $n$ robustnesses to calculate. In Table \ref{tab:affineEquality} we present the relevant values for the family of multicontrol-$T$ gates ($t=2$) and make two observations.  First, looking across the rows of Table \ref{tab:affineEquality}, notice that the values for fixed $k$ are constant with $n$, assuming $k \leq n$.  Indeed, this is a generic feature of the $M_{t,n}$ gates as formalised by Observation \ref{obs:M_t_n_dimension}. Second, looking down the last column of Table \ref{tab:affineEquality}, we see that up until $k=4$, $\robmag{M_{t,n}\ketcal{L}{}}{}$ increases with $\dim(\mathcal{L})$, but at $k=5$ the value drops. With a little thought we can see that this is necessarily the case if $\RChoi{M_{t,5}}< \Mcap{M_{t,5}}$; we saw earlier that for diagonal gates $U$ the Choi state robustness is equal to $\mathcal{R}\qty(U\ket{+}^{\otimes n})$, and $\ket{+}^{\otimes n}$ is a representative state for the $k=n$ case. 

Our current techniques limit us to five-qubit operations, so we are unable to confirm whether $\mathcal{R}(M_{t,n}\ketcal{L}{})$ continues to decrease with increasing $\dim(\mathcal{L})$. An intuition for why a decrease is plausible goes as follows. A stabiliser state $\ket{\mathcal{L}}$ with $\dim(\mathcal{L})=k$ will have $2^k$ equally weighted terms when written in the computational basis, so will have a normalisation factor of $2^{-k/2}$. The non-stabiliser state $M_{t,n}\ket{\mathcal{L}}$ is identical to $\ketcal{L}{}$ apart from the phase on the all-zero term $\ket{0\ldots0}$. As $k$ becomes large, the amplitude of the term $\frac{e^{i \pi /2^t}}{2^{k/2}} \ket{0\ldots0}$ becomes very small, so that $M_{t,n}\ket{\mathcal{L}}$ has high fidelity with the stabiliser state $\ketcal{L}{}$. We would therefore expect $M_{t,n}\ket{\mathcal{L}}$ to have a small robustness if $k$ is large.

\end{document}